\documentclass[lettersize,journal]{IEEEtran}
\usepackage[linesnumbered,ruled,vlined]{algorithm2e}
\usepackage{graphicx}
\usepackage{subcaption}
\usepackage{array}
\usepackage{caption}
\usepackage{textcomp}
\usepackage{stfloats}
\usepackage{url}
\usepackage{verbatim}
\usepackage{mathrsfs}
\usepackage{algpseudocode}
\usepackage{calligra}
\usepackage{cite}
\usepackage{amsmath,amssymb,amsfonts}
\usepackage{amsthm}
\usepackage{graphicx}
\usepackage{xcolor}
\usepackage{float}
\def\BibTeX{{\rm B\kern-.05em{\sc i\kern-.025em b}\kern-.08em
		T\kern-.1667em\lower.7ex\hbox{E}\kern-.125emX}}
\def\linspread{1}
\def\linspreadalgr{1}

\begin{document}
\captionsetup[figure]{name={Fig.},labelsep=period,font=footnotesize,justification=raggedright,singlelinecheck=off} 	
%
%
%

\title{Ambiguity Function Analysis and Optimization of Frequency-Hopping MIMO Radar with Movable Antennas}

\author{\IEEEauthorblockN{Xiang Chen,~\IEEEmembership{Student Member,~IEEE,} Ming-Min Zhao,~\IEEEmembership{Senior Member,~IEEE,} Min Li,~\IEEEmembership{Member,~IEEE,} Liyan Li,~\IEEEmembership{Member,~IEEE,} Min-Jian Zhao,~\IEEEmembership{Member,~IEEE,} and Jiangzhou Wang,~\IEEEmembership{Fellow,~IEEE}}
	\thanks{X. Chen, M. M. Zhao, M. Li, L. Li, and M. J. Zhao are with the College of Information Science and Electronic Engineering, Zhejiang University, Hangzhou 310027, China, and also with Zhejiang Provincial Key Laboratory of Multi-Modal Communication Networks and Intelligent Information Processing, Hangzhou 310027, China (e-mail: \{12231089, zmmblack, min.li, liyan\_li, mjzhao\}@zju.edu.cn). J. Wang is with the School of Information Science and Engineering, Southeast University, Nanjing 211119, China (e-mail: j.z.wang@seu.edu.cn).}
	}

\IEEEpubid{\begin{minipage}{\textwidth}\raggedright
	Copyright~\copyright~2025 IEEE. Personal use of this material is permitted. However, permission to use this material for any other purposes must be obtained from the IEEE by sending a request to pubs-permissions@ieee.org.
\end{minipage}}


\maketitle
\begin{abstract}
In this paper, we propose a movable antenna (MA)-enabled frequency-hopping (FH) multiple-input multiple-output (MIMO) radar system and investigate its sensing resolution. Specifically, we derive the expression of the ambiguity function and analyze the relationship between its main lobe width and the transmit antenna positions. In particular, the optimal antenna distribution to achieve the minimum main lobe width in the angular domain is characterized. We discover that this minimum width is related to the antenna size, the antenna number, and the target angle. Meanwhile, we present lower bounds of the ambiguity function in the Doppler and delay domains, and show that the impact of the antenna size on the radar performance in these two domains is very different from that in the angular domain. Moreover, the performance enhancement brought by MAs exhibits a certain trade-off between the main lobe width and the side lobe peak levels. Therefore, we propose to balance between minimizing the side lobe levels and narrowing the main lobe of the ambiguity function by optimizing the antenna positions. To achieve this goal, we propose a low-complexity algorithm based on the Rosen’s gradient projection method, and show that its performance is very close to the baseline. Simulation results are presented to validate the theoretical analysis on the properties of the ambiguity function, and demonstrate that MAs can reduce the main lobe width and suppress the side lobe levels of the ambiguity function, thereby enhancing radar performance.
\end{abstract}
	

\section{Introduction}
The emergence of  frequency hopping (FH) multiple-input multiple-output (MIMO) radar  has garnered significant interest from both industry and academia, thanks to its enhanced capabilities compared to other MIMO radar waveform modalities\cite{9656537}. In FH-MIMO radar, the total bandwidth is divided into many sub-bands, and only a subset of them are used 
at a time. The sub-bands used at each antenna randomly varies over time \cite{9540344}, which leads to some major advantages, such as better security and enhanced anti-jamming capability.
Specifically,  the FH technology leverages confidential FH patterns to identify available frequencies, thereby presenting a significant challenge for eavesdroppers to intercept FH-based radar signals  \cite{shi_enhanced_2022}, and it is difficult for conventional jammers to perform frequency tracking and inject interference on this kind of signal \cite{9266402}.
Besides, adopting the FH strategy incurs negligible sensing performance degradation, as compared to using full bandwidth signals \cite{163565}. Currently, the application of FH-MIMO radar has garnered considerable attention in the field of dual-function radar and communication (DFRC) or integrated sensing and communication (ISAC) system \cite{9427572}, as FH-MIMO radar can incorporate information into the fast-time subpulses within a pulse repetition interval (PRI), thereby achieving a high data rate\cite{9969893}. Therefore, FH-MIMO radar holds great promise in diverse applications such as enhancing target detection accuracy in complex electromagnetic environments, improving anti-interference and anti-interception capabilities in military surveillance, enabling high-resolution imaging for remote sensing, and contributing to more reliable and efficient wireless communication networks \cite{shi_enhanced_2022,9266402,163565,9427572,9969893}.\IEEEpubidadjcol

FH code and waveform design are two critical factors affecting the FH-MIMO radar and communication performance, thus they have been thoroughly investigated recently in the literature \cite{chenMIMORadarAmbiguity2008,gogineniFrequencyHoppingCodeDesign2012,zhouFrequencyhoppingCodeOptimization2016,hanJointlyOptimalDesign2016,zhangDualFunctionMIMORadarCommunications2021,eedaraOptimumCodeDesign2020,hassanienDualfunctionMIMORadarcommunications2017,wangPhaseModulatedCommunications2020}.
In  \cite{chenMIMORadarAmbiguity2008}, a simulated annealing algorithm was proposed to optimize the FH code in order to reduce the side lobe levels of the ambiguity function.
In \cite{gogineniFrequencyHoppingCodeDesign2012}, the FH code was optimized by reducing the block 
coherence measure of the sensing matrix with less computational complexity. However, most coherence values are too small, which cannot be measured accurately in practical applications.
Therefore, a design criterion based on FH radar coincidence imaging was investigated in \cite{zhouFrequencyhoppingCodeOptimization2016}, focusing on minimizing the difference between the correlation matrix of the dictionary matrix and an identity matrix. Moreover, a game theory framework was studied in \cite{hanJointlyOptimalDesign2016} to achieve better sensing performance via joint design of the FH code and amplitude matrices. As the deployment of FH-MIMO radar in DFRC systems becomes more prevalent, the integration of communication functionality  may adversely impact the radar's sensing capabilities. To mitigate these effects, several optimization methods were proposed to optimize the FH matrices in DFRC systems  \cite{zhangDualFunctionMIMORadarCommunications2021,eedaraOptimumCodeDesign2020,hassanienDualfunctionMIMORadarcommunications2017,wangPhaseModulatedCommunications2020}. 

In terms of transmit waveforms, the geometry of the antenna array is an important factor that determines the system performance. Recently, several sparse arrays have been proposed to increase the degrees of freedom (DoFs) of MIMO radar, e.g., the minimum redundancy arrays (MRAs) \cite{moffetMinimumredundancyLinearArrays1968}, coprime arrays \cite{vaidyanathanSparseSensingCoPrime2011} and nested arrays \cite{palNestedArraysNovel2010}, etc. In \cite{chenMinimumRedundancyMIMO2008}, the authors proposed to use an MRA in monostatic MIMO radar systems, and it was demonstrated that the MRA attains excellent performance in the rejection of main lobe interference. In \cite{qinDOAEstimationMixed2014}, a nested MIMO system was exploited  to estimate the directions of arrival (DOAs) of some uncorrelated and coherent targets, where a sparsely located uniform linear array (ULA) was deployed at the receiver. However, the DoFs provided by irregular antenna array was not fully utilized in the above mentioned works as the antenna spacing therein is fixed after being manufactured, which limits their capability to adapt to different task requirements.

{\color{black}With the increasing demand for communication capacity in the next-generation wireless communication systems, optimizing the geometry of the antenna array for better communication performance has also garnered widespread interests \cite{9264694,wongFluidAntennaMultiple2022,10653737,10694739,zhuMovableAntennasWireless2023,maCapacityMaximizationMovable2023,10497534,10709885}.
	For example,  a fluid antenna system (FAS) was proposed to further explore the DoFs in the spatial domain by changing the antenna position more flexibly \cite{9264694,wongFluidAntennaMultiple2022,10653737,10694739}. By using conductive fluids as materials for antennas, the receive antenna position can be switched freely among all candidate ports over a fixed-length line, and thus the signal with the highest signal-to-noise ratio (SNR) can be received \cite{9264694}. A fluid antenna multiple access (FAMA) scheme was proposed in \cite{wongFluidAntennaMultiple2022} to support multiple transceivers with a single fluid antenna at each mobile user. By selecting the positions of the fluid antennas at different users, the favorable channel condition
	with mitigated interference can be obtained. 
	In \cite{10653737}, the design of position index modulation (PIM) for FAS was proposed to decrease the bit error rate (BER) while taking advantage of the rate gain in index modulation. The performance of physical layer security (PLS) in fluid antenna-aided communication systems under arbitrary correlated fading channels was investigated in \cite{10694739}, and it demonstrated that applying fluid antenna to physical layer security can offer higher security and reliability compared to traditional antenna systems. In summary,
	it was shown in \cite{9264694,wongFluidAntennaMultiple2022,10653737,10694739} that changing the antenna positions can efficiently improve the wireless channel capacity and physical layer security. }

Recently, a new MIMO communication system enabled by movable antennas (MAs) was proposed in \cite{zhuMovableAntennasWireless2023} and \cite{maCapacityMaximizationMovable2023} to reshape the MIMO channel matrix between the transceivers, such that more spatial DoFs can be exploited for enhancing the channel capacity. Specifically, by connecting the MAs to RF chains via flexible cables, the MA positions can be adjusted by controllers in real time,
such as stepper motors or servos \cite{hejresNullSteeringPhased2004,zhuravlevExperimentalSimulationMultistatic2015,basbugDesignSynthesisAntenna2017}. In \cite{10497534}, a general channel estimation framework for MA systems by exploiting the multi-path field response channel structure was proposed that can estimate the complete CSI between the Tx and Rx regions with a high accuracy. The MA-aided wideband communications employing orthogonal frequency division multiplexing (OFDM) in frequency-selective fading channels was investigated in \cite{10709885}, which outperformed conventional systems with fixed-position antennas (FPAs) under the wideband channel setup.
 Using MA array to enhance sensing capabilities in wireless sensing and ISAC systems has also garnered significant attention \cite{ma_movable_2024,10696953,khalili_advanced_2024}.  A wireless sensing system using MA array was proposed in \cite{ma_movable_2024}, where the authors proposed to optimize the  antenna positions to minimize the Cramér-Rao bound (CRB). In \cite{10696953},  a joint beamforming design and MA position  optimization problem  was considered, and an alternating optimization (AO) algorithm that combines  semidefinite relaxation (SDR) and successive convex approximation (SCA) techniques  was proposed to tackle this problem. In \cite{khalili_advanced_2024}, the authors investigated the resource allocation problem in ISAC systems exploiting MAs, and chance constraints were introduced and integrated into the sensing quality of service (QoS) framework to precisely control the impact of dynamic radar cross-section (RCS) variations. 

Although there have been numerous studies on the application of MAs in the field of communication and wireless sensing dominated by communication waveforms, the use of MA array in FH-MIMO radar systems remains relatively scarce, and it is not clear how MAs can improve its sensing performance and to what extent. Motivated by the above, we propose an MA-enabled FH-MIMO radar system in this paper and investigate the optimization of the MAs' positions to enhance the performance of the ambiguity function. The main contributions of this paper are summarized as follows:
\begin{itemize}
	\item First, to evaluate the potential of MA in improving the MIMO radar resolutions, we analyze the relationship between the antenna positions and main lobe width of the ambiguity function in the angular domain. The optimal antenna distribution for achieving the minimum main lobe width is identified, which is shown to be related to the antenna size, the antenna number and the target angle. Different from the angular domain, we present lower bounds of the ambiguity function in the Doppler and delay domains, which unveil that when the antenna size is sufficiently large, further increasing it will not provide additional radar performance gain. These analytical results demonstrate that compared to the Doppler and delay domains, optimizing the MA positions exhibits a more remarkable enhancement on the radar resolution performance in the angular domain.
	\item Second, we propose to achieve a good tradeoff between main lobe width and side lobe peak levels for the ambiguity function in these three domains by optimizing the MA positions, and an optimization problem is formulated accordingly. The considered problem is highly non-convex and thus very difficult to solve. To tackle this problem, we propose an efficient algorithm based on the Rosen's gradient projection method (RGPM), which is able to achieve comparable performance as the baseline but with much lower computational complexity.
	\item Third,  we show that  the proposed system has two major advantages over traditional FH-MIMO radar systems. First, the proposed system, with the aid of MAs, achieves an ambiguity function with lower side lobe levels and narrower main lobe width in the angular, Doppler and delay domains, indicating enhanced sensing resolution. Second, by flexibly adjusting the antenna positions, the radar can achieve improved performance in various tasks, e.g., if better resolution of target angles (velocity or distance) is required, the antenna positions can be adjusted to improve the performance of the ambiguity function in the angular (Doppler or delay) domain.
\end{itemize} 

The remainder of this paper is organized as follows. Section II describes the considered MA-enabled FH-MIMO radar system model and the corresponding ambiguity function. In Section III, we analyze the relationship between the ambiguity function and the antenna positions. In Section IV, we propose a low complexity algorithm to solve the antenna position optimization problem. Section V provides the numerical results and discussions. Finally, we conclude this paper in Section VI. 

\emph{Notations}: Scalars, vectors and matrices are respectively denoted by lower/upper case, boldface lower case and boldface upper case letters. For an arbitrary matrix $\mathbf{A}$, $\mathbf{A}^T$, $\mathbf{A}^*$ and $\mathbf{A}^H$ denote its transpose, conjugate and conjugate transpose respectively. $\|\cdot\|$ denotes the Euclidean norm of a complex vector, and $\lvert\cdot\lvert$ denotes the absolute value of a complex scalar. $\lceil \cdot \rceil$ represents the round-up operator. For a complex number $x$, $\Re\{x\}$ denotes its real part and $\angle x$ denotes its angle. $\mathbf{I}$ and $\mathbf{0}$ denote an identity matrix and an all-zero vector with appropriate dimensions, respectively. $\mathbb{C}^{n\times m}$ denotes the space of $n\times m$ complex matrices. 

\section{SIGNAL MODEL}
\subsection{MA-Enabled FH-MIMO Radar System}
As shown in Fig. \ref{figure:fig.1}, we consider an MA-enabled FH-MIMO radar system equipped with a colocated transmitter and receiver, which are comprised of linear arrays with $M_t$ and $M_r$ antennas, respectively. Since the transmit antenna array is the primary factor that affects the ambiguity function, we assume that the transmit antennas are movable while the receive antenna positions are fixed with half wavelength spacing for ease of analysis \cite{chenMIMORadarAmbiguity2008}.
\begin{figure*}[!ht]
	\centering
	\includegraphics[width=14cm]{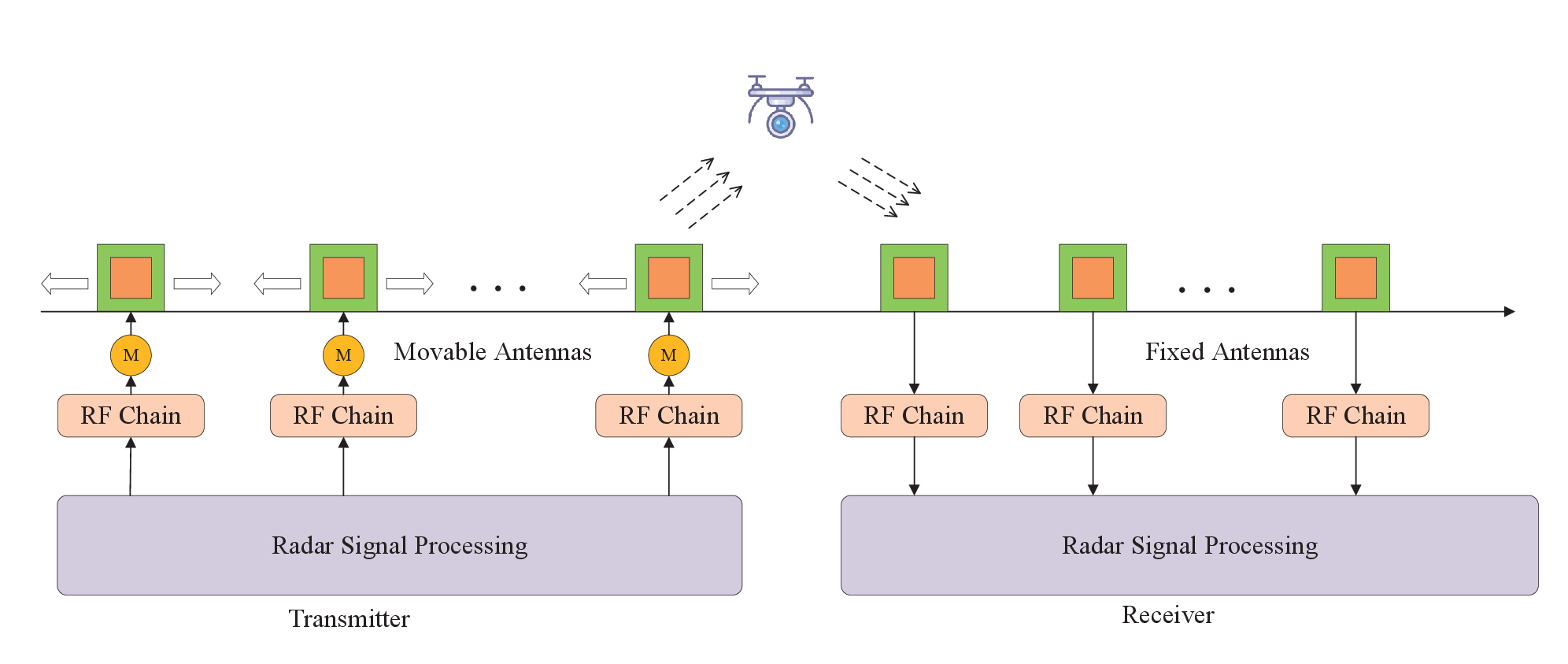}
	\caption{ Proposed MA-enabled FH-MIMO radar system.}
	\label{figure:fig.1}		
	\vspace{-20pt}
\end{figure*} 
Each MA is attached to an electrical machinery, such that the interval between two adjacent antennas can be dynamically adjusted \cite{maCapacityMaximizationMovable2023}. Let $d_{t,i}$ $(d_{r,i})$, $1\leq i\leq M_t-1$ $(1\leq i\leq M_r-1)$ denote the interval between the $(i-1)$-th and $i$-th transmit (receive) antenna and define  $d_{t,0}=0$, $d_{r,0}=0$, $d_{r,i}=\frac{\lambda}{2},i\in\mathcal{M}_r\triangleq[1,2,\cdots,M_r-1]$. Then, the transmit antenna position vector can be denoted by $\mathbf{x}_{t}=[x_{t,0},x_{t,1},\cdots,x_{t,M_t-1}]^T$, where $x_{t,0}=0$ and  $x_{t,m}=\sum_{i=0}^{m}d_{t,i},m\in \mathcal{M}_t\triangleq[1,2,\cdots,M_t-1]$. Similarly, the receive antenna position vector can be expressed as $\mathbf{x}_{r}=[x_{r,0},x_{r,1},\cdots,x_{r,M_r-1}]^T$, where $x_{r,0}=0$ and $x_{r,m}=\frac{\lambda}{2}m,m\in\mathcal{M}_r$. Accordingly, the steering vectors of the transmit and receive antenna arrays are  respectively given by $	\mathbf{a}(\mathbf{x}_t,\alpha)=[1,e^{j\frac{2\pi}{\lambda}x_{t,1}\sin\alpha},\cdots,e^{j\frac{2\pi}{\lambda}x_{t,M_t-1}\sin\alpha}]^T$ and $\mathbf{b}(\mathbf{x}_r,\alpha)=[1,e^{j\pi\sin\alpha},\cdots,e^{j\pi(M_r-1)\sin\alpha}]^T$,
 where $\lambda$ denotes the signal wavelength and $\alpha$ is the steering angle of the array.

Consider a target at $(\hat{\tau},\hat{v},\theta)$, where $\hat{\tau}$ denotes the delay corresponding to the target range, $\hat{v}$ is the Doppler frequency of the target and $\theta\in[-\frac{\pi}{2},\frac{\pi}{2}]$ represents the direction angle of the target. Then, the received signal can be represented by
\begin{equation}\label{eq:0412_3}
	\small
	\mathbf{y}_{\hat{\tau},\hat{v},\theta}(t)= \mathbf{a}(\mathbf{x}_t,\theta)^T\boldsymbol{\phi}(t-\hat{\tau})\mathbf{b}(\mathbf{x}_r,\theta)e^{j2\pi \hat{v}t}+\mathbf{n}(t),
\end{equation}
where $\boldsymbol{\phi}(t)=[\phi_0(t),\phi_1(t),\cdots,\phi_{M_t-1}(t)]^T$, $\boldsymbol{n}(t)=[n_0(t),n_1(t),\cdots,n_{M_r-1}(t)]^T$, $\phi_m(t)$ represents the FH waveform transmitted from the $m$-th transmit antenna and $n_m(t)$ denotes the Gaussian noise received by the $m$-th receive antenna.

As the FH waveform, the pulse width $T_w$ is divided into $Q$ sub-pulses of width $\Delta_t=T_w/Q$ each \cite{9681340}. Therefore, the $m$-th FH waveform during each pulse can be further expressed as\cite{chenMIMORadarAmbiguity2008}
\begin{equation}\label{eq:2}
	\small
	\phi_m(t)=\sum_{q=0}^{Q-1}e^{j2{\pi}{c_{m,q}}{\Delta}_ft}s(t-q{\Delta}_t),
\end{equation}
where $c_{m,q} \in \mathbb{K}$ is the FH code with $\mathbb{K} \triangleq \{1,2,\cdots,K\}$ being the set of available hops, $\Delta_f$ represents the frequency hopping interval and $s(t)$ represents the pulse function which is defined as
\begin{equation}\label{eq:3}
	\small
	s(t) = \left\{
	\begin{aligned}
		&1,\quad 0<t<\Delta_t, \\
		&0,\quad{\mathrm{otherwise}}.
	\end{aligned}
	\right.
\end{equation}
Note that the waveforms in an FH-MIMO radar system  are required to be orthogonal for zero Doppler and zero delay (see \cite{chenMIMORadarAmbiguity2008}), thus the condition
$c_{m,q}\neq c_{m',q}, \forall q, m\neq m'$
must be satisfied during each sub-pulse that comprises the radar pulse. This implies that the transmit antenna number $M_t$ that can be employed is upper bounded by the hop number $Q$. In this paper, our main focus is to investigate the radar performance enhancement brought by MA, and the FH code is designed by adopting the method presented in \cite{eedaraOptimumCodeDesign2020}.

\subsection{Ambiguity Function}
\vspace{-0.1cm}
For conventional single-input multiple-output (SIMO) radar systems, the following ambiguity function is usually employed to characterize the radar resolution performance\cite{levanon2004radar}:
\begin{equation}\label{eq:0414_1}
	\small
	\lvert\chi(\tau,v)\lvert\triangleq\bigg\lvert\int_{-\infty}^{\infty}u(t)u^*(t+\tau)e^{2\pi vt}dt\bigg\lvert,
\end{equation}
where $u(t)$ denotes the radar waveform. This two-dimensional function indicates the matched filter output at the receiver in the presence of Doppler mismatch $v$ and delay mismatch $\tau$. Such a concept is then extended to MIMO radar systems in\cite{sanantonioMIMORadarAmbiguity2007a}. Based on these existing works, we introduce the definition of ambiguity function in the considered MA-enabled FH-MIMO radar system in the following. 

Specifically, the receiver is designed to capture the target signal using a matched filter with assumed parameters $(\hat{\tau}',\hat{v}',\theta')$, where $\theta'\in[-\frac{\pi}{2},\frac{\pi}{2}]$. The reference signal of the matched filter is defined as $\mathbf{g}_{\hat{\tau}',\hat{v}',\theta'}(t)\triangleq\mathbf{a}(\mathbf{x}_t,\theta')^T\boldsymbol{\phi}(t-\hat{\tau}')\mathbf{b}(\mathbf{x}_r,\theta')e^{j2\pi \hat{v}'t}$, thus based on the received signal model given in \eqref{eq:0412_3}, the matched filter output can be expressed as
\begin{equation}	\label{eq:12}
	\small
	\begin{aligned}
		&\int_{-\infty}^{+\infty}\mathbf{g}_{\hat{\tau}',\hat{v}',\theta'}^H(t)\mathbf{y}_{\hat{\tau},\hat{v},\theta}(t)dt
		\\&=\begin{aligned}[t]
			&\bigg(\sum_{n=0}^{M_r-1}e^{j2\pi(\sum_{i=0}^{m}d_{r,i}\sin\theta-\sum_{i=0}^{m'}d_{r,i}\sin\theta')}\bigg)\\&T(\hat{\tau},\hat{\tau}',\hat{v},\hat{v}',\theta,\theta')+\int_{-\infty}^{+\infty}\mathbf{g}_{\hat{\tau}',\hat{v}',\theta'}^H(t)\mathbf{n}(t)dt,
		\end{aligned} 
	\end{aligned}
\end{equation}
where
\begin{equation}\label{eq:14}
	\small
	\begin{aligned}
		&T(\hat{\tau},\hat{\tau}',\hat{v},\hat{v}',\theta,\theta')\\&\triangleq\begin{aligned}[t]\sum_{m=0}^{M_t-1}\sum_{m'=0}^{M_t-1}
			&\int_{-\infty}^{+\infty}\phi_m(t-\hat{\tau})\phi^*_{m'}(t-\hat{\tau}')e^{j2\pi(\hat{v}-\hat{v}')t}dt\\&e^{j2\pi(\sum_{i=0}^{m}d_{t,i}\sin\theta-\sum_{i=0}^{m'}d_{t,i}\sin\theta')/\lambda}.
		\end{aligned}
	\end{aligned}
\end{equation}
In \eqref{eq:12},  $\mathbf{y}_{\hat{\tau},\hat{v},\theta}(t)$ represents the received signal reflected from the target at $(\hat{\tau},\hat{v},\theta)$, $\mathbf{g}_{\hat{\tau}',\hat{v}',\theta'}(t)$ represents the reference signal of the matched filter, and $\int_{-\infty}^{+\infty}\mathbf{g}_{\hat{\tau}',\hat{v}',\theta'}^H(t)\mathbf{y}_{\hat{\tau},\hat{v},\theta}(t)dt$ represents the coherent integration process between the received and  reference signals in the matched filter. The first term on the right-hand side of \eqref{eq:12} denotes the coherent
integration result between the echo  and  reference
signals, while $\int_{-\infty}^{+\infty}\mathbf{g}_{\hat{\tau}',\hat{v}',\theta'}^H(t)\mathbf{n}(t)dt$ denotes the coherent integration result between the reference signal and the noise. In \eqref{eq:14}, $\phi_m(t-\hat{\tau})e^{j2\pi\hat{v}t}dte^{j2\pi(\sum_{i=0}^{m}d_{t,i}\sin\theta)/\lambda}$ represents the echo signal transmitted by the 
$m$-th transmit antenna and reflected by the target at $(\hat{\tau},\hat{v},\theta)$, and $\phi^*_{m'}(t-\hat{\tau}')e^{j2\pi\hat{v}'t}dte^{j2\pi(\sum_{i=0}^{m'}d_{t,i}\sin\theta')/\lambda}$ represents the reference signal transmitted by the ${m'}$-th transmit antenna with assuming parameters $(\hat{\tau}',\hat{v}',\theta')$. From \eqref{eq:12} and \eqref{eq:14}, we can observe that there is a strong coupling effect between the radar's transmit MA array and the radar waveform, while the receive antenna array appears as an independent multiplier term that is in general unrelated to the radar waveform. Therefore, we define the ambiguity function in the proposed MA-enabled FH-MIMO radar system as
\begin{equation}\label{eq:1208_1}
	\small
	\begin{aligned}
		&\chi(\tau,v,\theta,\theta')\\&\triangleq\begin{aligned}[t]\sum_{m=0}^{M_t-1}\sum_{m'=0}^{M_t-1}
			&\int_{-\infty}^{+\infty}\phi_m(t)\phi^*_{m'}(t+\tau)e^{j2\pi vt}dt\\&e^{j2\pi(\sum_{i=0}^{m}d_{t,i}\sin\theta-\sum_{i=0}^{m'}d_{t,i}\sin\theta')/\lambda},
		\end{aligned}
	\end{aligned}
\end{equation}
 which indicates how the waveforms $\phi_m(t)$'s and the transmit MA array affect the angular, Doppler and range resolutions.\footnote{As observed from \eqref{eq:1208_1}, the angular ambiguity is related to $d_{t,i}\sin\theta$ and $d_{t,i}\sin\theta'$, where $i=0,1,\cdots,M_t-1$. This implies that it cannot be simply characterized by the spatial frequency mismatch $\delta_f$, where $\delta_f \triangleq \sin\theta - \sin\theta'$. Consequently, we retain the use of $\theta$ and $\theta'$ to describe the angular ambiguity throughout this work.} 
Since the ambiguity function represents the radar's ability to distinguish between echo signals under two sets of target parameters in the absence of noise \cite{sanantonioMIMORadarAmbiguity2007a}, the impact of noise on the matched filter output, i.e., $\int_{-\infty}^{+\infty}\mathbf{g}_{\hat{\tau}',\hat{v}',\theta'}^H(t)\mathbf{n}(t)dt$, has been neglected in \eqref{eq:1208_1}.
Then, by substituting the FH waveform in \eqref{eq:2} into $\chi(\tau,v,\theta,\theta')$, we can obtain the following new form of the ambiguity function:
\begin{equation}\label{eq:15}
	\small
	\begin{aligned}
		&\chi(\tau,v,\theta,\theta')\\
		&=\!\!\!\!
		\begin{aligned}[t]\sum_{m,m'\!=0}^{M_t\!-\!1}\sum_{q,q'\!=0}^{Q-1}&\chi^{r}(\tau\!-\!(q'\!-\!q){\Delta}_t,v\!-\!(c_{m'\!,q'}-c_{m\!,q})\Delta_f)\\&e^{j2{\pi}{\Delta}_f(c_{m,q}-c_{m',q'})q{\Delta}_t}e^{-j2{\pi\Delta}_fc_{m',q'}\tau}\\&e^{j2\pi(\sum_{i=0}^{m}d_{t,i}\sin(\theta)-\sum_{i=0}^{m'}d_{t,i}\sin(\theta'))/\lambda},
		\end{aligned}
	\end{aligned}
\end{equation} 
where
\begin{equation}\label{eq:16}
	\small
	\begin{aligned}
		&\chi^{r}(\tau,v) \triangleq \! \int_{0}^{{\Delta}_t}s(t)s(t+\tau)e^{j2{\pi}vt}dt\\
		&=\!  \left\{\begin{aligned}
			&\frac{{\Delta}_t\! -\!\!  \lvert\tau\lvert}{{\Delta}_t}e^{j{\pi}v(\Delta_t -\! \tau)}{\rm sinc}(v(\Delta_ t\! -\! \lvert\tau\lvert)),\quad \lvert\tau\lvert<\Delta_t,\\
			&0,\quad\quad\quad\quad\quad\quad\quad\quad\quad\quad\quad\quad\quad\quad\quad \mathrm{otherwise}.
		\end{aligned}
		\right.
	\end{aligned}
\end{equation}
It is observed from \eqref{eq:15} that the ambiguity function $\chi(\tau,v,\theta,\theta')$ is in the form of the sum of several complex numbers, and the antenna positions will affect the phase of each complex number, thereby influencing the main lobe width and side lobe levels, which will ultimately reflect on the radar performance. Besides, from \eqref{eq:16}, we can see that the FH code not only affects the phase of each complex number but also its amplitude.
Thus, both the antenna distribution and FH code will impact the ambiguity function. However, as the influence of the FH-code on the radar system performance has been investigated before in \cite{chenMIMORadarAmbiguity2008,eedaraOptimumCodeDesign2020} while the relationship between the ambiguity function and antenna distribution is still not clear. Therefore, in this paper, we will focus our discussion on the latter in this work. The FH code design follows the approach in \cite{eedaraOptimumCodeDesign2020}.
\section{ambiguity function analysis}\label{section3}
In this section, we theoretically analyze the relationship between $\lvert\chi(\tau,v,\theta,\theta')\lvert$ and the transmit antenna positions in $\mathbf{d}=[d_{t,1},\cdots,d_{t,M_t-1}]^T$ across the angular, Doppler and delay domains. First, we identify the optimal antenna distribution that minimizes the main lobe width of the ambiguity function in the angular domain. A closed-form expression of the minimum main lobe width is presented as a function of the target angle, antenna size and  antenna number. Second, we observe that the Doppler frequency offset and time delay error affect both the phase and amplitude of the terms in $\chi(\tau,v,\theta,\theta')$, which is very different from the angle error that only affect the phase of them. This poses significant challenges to obtain the optimal antenna distribution that achieves the minimum main lobe width in the Doppler and delay domains. Therefore, we investigate the lower bounds of the ambiguity function in these two domains, and some useful insights are obtained accordingly.
\subsection{Angular Domain}\label{section3-A}
In order to analyze the characteristics of the ambiguity function in the angular domain, we assume that the target’s distance and velocity are accurately estimated, i.e., $\tau=0$ and $v=0$, and focus on $\chi(0,0,\theta,\theta')$ which can be expressed as
\begin{equation}\label{eq:0928_1}
	\small
		\begin{aligned}
			&\chi(0,0,\theta,\theta')\\
			&=\!\!\!\!
			\begin{aligned}[t]\sum_{m,m'\!=0}^{M_t\!-\!1}\sum_{q,q'\!=0}^{Q-1}&\chi^{r}(0\!-\!(q'\!-\!q){\Delta}_t,0\!-\!(c_{m'\!,q'}-c_{m\!,q})\Delta_f)\\&e^{j2{\pi}{\Delta}_f(c_{m,q}-c_{m',q'})q{\Delta}_t}\\&e^{j2\pi(\sum_{i=0}^{m}d_{t,i}\sin(\theta)-\sum_{i=0}^{m'}d_{t,i}\sin(\theta'))/\lambda}.
			\end{aligned}
		\end{aligned}
\end{equation}
As can be seen from \eqref{eq:0928_1} the phase information of the last exponential term is a linear combination of $d_{t,i}\sin\theta$ and $d_{t,i}\sin\theta'$, $i\in\mathcal{M}_t$, thus, we can infer that the antenna distribution has a strong impact on the angular resolution of the considered system.  Further, based on the orthogonality of the FH-MIMO radar waveform, the angular domain ambiguity function in \eqref{eq:0928_1}  can be simplified to
 \begin{equation}\label{eq:20}
 	\small
	 	\begin{aligned}
		 		&\chi(0,0,\theta, \theta')= \sum_{m=0}^{M_t-1}e^{j2\pi(\sin \theta-\sin\theta')\sum_{i=0}^{m}d_{t,i/\lambda}}.
		 	\end{aligned}
\end{equation}
According to the definition of the ambiguity function, narrower main lobe width usually implies higher resolutions of the radar system. Therefore, we aim to derive the optimal antenna distribution that achieves the minimum main lobe width $B_{min}$ of $\lvert\chi(0,0,\theta, \theta')\lvert$, and investigate the impacts of the transmit antenna number $M_t$, the transmit antenna size $L$ and the target direction angle $\theta$ on $B_{min}$. The main results are given in the following theorem.
\newtheorem{theorem}{\bf Theorem}
\begin{theorem}\label{thm1}
For a given $L$, letting $\tau =0,v=0$, the optimal antenna distribution $\mathbf{d}^*$ that achieves the minimum main lobe width (defined as the distance between the first two null points) is given by
\begin{equation}\label{eq:0406_2}
	\small
	d_{t,i}=\left\{\begin{aligned}
		&{\lambda}/{2},\quad i=1,2,\cdots,\bigg\lceil(M_t-2)/{2}\bigg\rceil\\
		&{\lambda}/{2},\quad i=\bigg\lceil{M_t}/{2}\bigg\rceil+1,\bigg\lceil{M_t}/{2}\bigg\rceil+2,\cdots,M_t-1\\
		&L-(M_t-2){\lambda}/{2},\quad i= \bigg\lceil{M_t}/{2}\bigg\rceil\\
	\end{aligned}
	\right
	..
\end{equation}
 This minimum main lobe width has a closed-form expression, which is 
\begin{equation}\label{eq:21}
	\small
	\begin{aligned}
		B_{min} =& \arcsin\bigg(\sin\theta+\frac{2}{4L/\lambda-M_t+2}\bigg)\\&-\arcsin\bigg(\sin\theta-\frac{2}{4L/\lambda-M_t+2}\bigg).
	\end{aligned}
\end{equation} 
\end{theorem} 
\begin{proof}
	
Let $\Delta\theta_{a}$ denote the angle error within the main lobe width, and  define $\chi_\mathbf{d}( \theta,\Delta\theta_{a})$ as the ambiguity function in the angular domain with the antenna distribution $\mathbf{d}$, i.e.,
\begin{equation}\label{eq:1020_1}
	\small
		\begin{aligned}
		&\lvert\chi_\mathbf{d}(\theta,\Delta\theta_{a})\lvert\triangleq \bigg\lvert\sum_{m=0}^{M_t-1}a_{m}(\mathbf{d},\Delta\theta_{a})\bigg\lvert,
	\end{aligned}
\end{equation}
where $a_m(\mathbf{d},\Delta\theta_{a})\triangleq e^{j2\pi\psi_m(\mathbf{d},\Delta\theta_{a})},m=0,1,\cdots,M_t-1$
and
\begin{equation}\label{eq:0506_2}
	\small
	\psi_{m}(\mathbf{d},\Delta\theta_{a})\triangleq 2\pi(\sin\theta-\sin(\theta+\Delta\theta_{a}))\sum_{i=0}^{m}d_{t,i}/\lambda.
\end{equation}
To prove that $\mathbf{d}^*$ achieves the minimum main lobe width, it is equivalent to demonstrating that for any antenna distribution $\mathbf{d}$, the inequality $\lvert\chi_\mathbf{d}(\theta,\Delta\theta_{a})\lvert\geq\lvert\chi_{\mathbf{d}^*}(\theta,\Delta\theta_{a})\lvert$ holds.
Due to the symmetry of $\lvert\chi_\mathbf{d}(\theta,\Delta\theta_{a})\lvert$, we will focus on the case $\Delta\theta_a\leq0$ first, while the detailed proof for the case $\Delta\theta_a>0$ is omitted for brevity. Besides, as the proof process depends on whether $M_t$ is even or odd, we will divide it into two cases based on the parity of $M_t$.

For the case when $M_t$
is even, we further divide the proof process into the following two steps: (1) since $\chi_\mathbf{d}(\theta,\Delta\theta_{a})$ can be viewed as the summation of $M_t$ unit vectors with different directions in the two-dimensional space, we propose to construct a proper Cartesian coordinate system (CCS) for ease of comparing the magnitudes of $\chi_\mathbf{d}(\theta,\Delta\theta_{a})$ and $\chi_{\mathbf{d}^*}(\theta,\Delta\theta_{a})$, then we obtain the x-axis and y-axis components of $\chi_{\mathbf{d}}(\theta,\Delta\theta_a)$, denoted by $\chi_\mathbf{d}(\theta,\Delta\theta_{a})_x$ and $\chi_\mathbf{d}(\theta,\Delta\theta_{a})_y$, respectively, as well as $\chi_{\mathbf{d}^*}(\theta,\Delta\theta_a)_x$ and  $\chi_{\mathbf{d}^*}(\theta,\Delta\theta_a)_y$; (2) we prove that for any antenna distribution $\mathbf{d}$, the inequalities $\lvert\chi_{\mathbf{d}}(\theta,\Delta\theta_a)_x\lvert\geq\lvert\chi_{\mathbf{d}^*}(\theta,\Delta\theta_a)_x\lvert$ and $\lvert\chi_{\mathbf{d}}(\theta,\Delta\theta_a)_y\lvert\geq\lvert\chi_{\mathbf{d}^*}(\theta,\Delta\theta_a)_y\lvert$ hold, which further lead to $\lvert\chi_{\mathbf{d}}(\theta,\Delta\theta_a)\lvert\geq\lvert\chi_{\mathbf{d}^*}(\theta,\Delta\theta_a)\lvert$. 

In the first step, since the choice of the origin and y-axis direction in the CCS affects the analysis of  $\lvert\chi_{\mathbf{x}}(\theta,\Delta\theta_a)_w\lvert$, $\mathbf{x}\in\{\mathbf{d},\mathbf{d}^*\}$, $w\in\{x,y\}$. We propose to select an appropriate origin and positive y-axis direction such that 
$\lvert\chi_{\mathbf{d}^*}(\theta,\Delta\theta_a)_x\lvert=0$ is satisfied. This allows us to simplify the analysis and we only need to compare
$\lvert\chi_{\mathbf{d}}(\theta,\Delta\theta_a)_y\lvert$ and $\lvert\chi_{\mathbf{d}^*}(\theta,\Delta\theta_a)_y\lvert$. To determine such an origin and direction, we analyze the key factor that affects the value of
 $\lvert\chi_{\mathbf{d}^*}(\theta,\Delta\theta_a)\lvert$, i.e., $\psi_m(\mathbf{d}^*,\Delta\theta_{a})$. Since $\psi_m(\mathbf{d}^*,\Delta\theta_{a})$ is a special case of $\psi_m(\mathbf{d},\Delta\theta_{a})$, it suffices to analyze $\psi_m(\mathbf{d},\Delta\theta_{a})$. Then, due to the facts that $d_{t,i}\geq\frac{\lambda}{2},i\in\mathcal{M}_t$ and $\sum_{i=0}^{M_t-1}d_{t,i}\leq L$, it is straightforward to see that  $\psi_m(\mathbf{d},\Delta\theta_{a}), m=0,1,\cdots,M_t-1$ should satisfy  $\psi_m(\mathbf{d},\Delta\theta_{a})\in[\psi^{min}_{m}(\Delta\theta_{a}),\psi^{max}_{m}(\Delta\theta_{a})]$, where
 \begin{equation}\label{eq:1017_7}
 	\small
 	\psi^{min}_{m}(\Delta\theta_{a})=m\pi(\sin\theta-\sin(\theta+\Delta\theta_{a}))
 \end{equation} 
 and
 \begin{equation}\label{eq:1017_8}
 	\small
 	\begin{aligned}
 		&\psi^{max}_{m}(\Delta\theta_{a})\\&=2\pi(\sin\theta-\sin(\theta+\Delta\theta_{a}))(\frac{L}{\lambda}-\frac{M_t-1-m}{2}),
 	\end{aligned}
 \end{equation} 
 which represents the minimum and maximum achievable values of $\psi_m$, respectively. By substituting $\mathbf{d}^*$ into \eqref{eq:1017_8}, we have
\begin{equation}\label{eq:1020_2}
	\small
	\begin{aligned}
		&\psi_m(\mathbf{d}^*,\Delta\theta_{a})\\&=\left\{\begin{aligned}
			&\psi_m^{min}(\Delta\theta_{a}),\quad m=0,1,\cdots,({M_t-2})/{2}\\
			&\psi_m^{max}(\Delta\theta_{a}),\quad m={M_t}/{2},{M_t}/{2}+1,\cdots,M_t-1\\
		\end{aligned}
		\right
		..
	\end{aligned}
\end{equation}
Next, based on \eqref{eq:1017_7}, \eqref{eq:1017_8} and \eqref{eq:1020_2}, it follows that 
\begin{equation}\label{eq:1021_7}
	\small
	\begin{aligned}
		&\psi_m(\mathbf{d^*},\Delta\theta_{a})+\psi_{M_t-m}(\mathbf{d^*},\Delta\theta_{a})\\&=\psi_{M_t/2}(\mathbf{d^*},\Delta\theta_{a})+\psi_{M_t/2-1}(\mathbf{d^*},\Delta\theta_{a}),
	\end{aligned}
\end{equation}
which implies that under the antenna distribution $\mathbf{d}^*$, $a_m(\mathbf{d}^*,\Delta\theta_{a})$ and  $a_{M_t-m-1}(\mathbf{d}^*,\Delta\theta_{a}),m=0,1,\cdots,M_t-1$ are symmetric with respect to $y_1\triangleq a_{M_t/2}(\mathbf{d}^*,\Delta\theta_{a})+a_{M_t/2-1}(\mathbf{d}^*,\Delta\theta_{a})$. Therefore, we can construct a CCS using the common intersection of $a_m(\mathbf{d}^*,\Delta\theta_{a}),m=0,1,\cdots,M_t-1$ as the origin and the direction of $y_1$ as the positive y-axis, as shown in Fig. \ref{figure:fig.1020} (a). As  $a_m(\mathbf{d}^*,\Delta\theta_{a})$ and  $a_{M_t-m-1}(\mathbf{d}^*,\Delta\theta_{a}),m=0,1,\cdots,M_t-1$ are symmetric with respect to the y-axis, we have
\begin{equation}\label{eq:1108_1}
	\small
	a_m(\mathbf{d}^*,\Delta\theta_{a})_x+a_{M_t-m-1}(\mathbf{d}^*,\Delta\theta_{a})_x=0,
\end{equation}
where $a_m(\mathbf{d}^*,\Delta\theta_{a})_x$ and $a_{M_t-m-1}(\mathbf{d}^*,\Delta\theta_{a})_x$ represents the x-axis components of $a_m(\mathbf{d}^*,\Delta\theta_{a})$ and $a_{M_t-m-1}(\mathbf{d}^*,\Delta\theta_{a})$, respectively.
 This implies that $\lvert\chi_{\mathbf{d}^*}(\theta,\Delta\theta_a)_x\lvert=0$. Therefore, $\lvert\chi_{\mathbf{d}}(\theta,\Delta\theta_a)_x\lvert\geq\lvert\chi_{\mathbf{d}^*}(\theta,\Delta\theta_a)_x\lvert$ holds for any $\mathbf{d}$.

\begin{figure}[t]
	\centering
	\begin{subfigure}[b]{0.4\textwidth}
		\vspace{-1pt}
		\centering{\includegraphics[width=\textwidth]{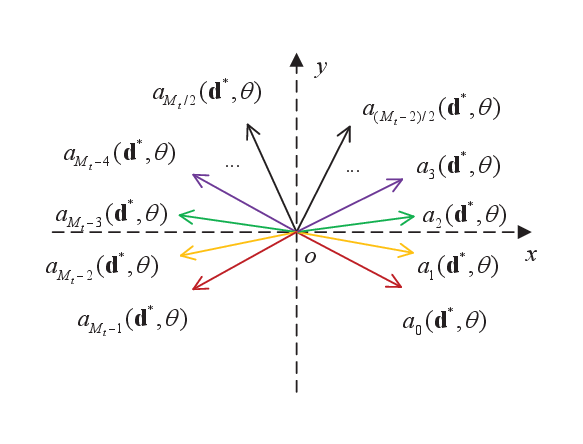}}
			\centering{\small(a)}
	\end{subfigure}
	\begin{subfigure}[b]{0.4\textwidth}
		\vspace{-1pt}
		\centering{\includegraphics[width=\textwidth]{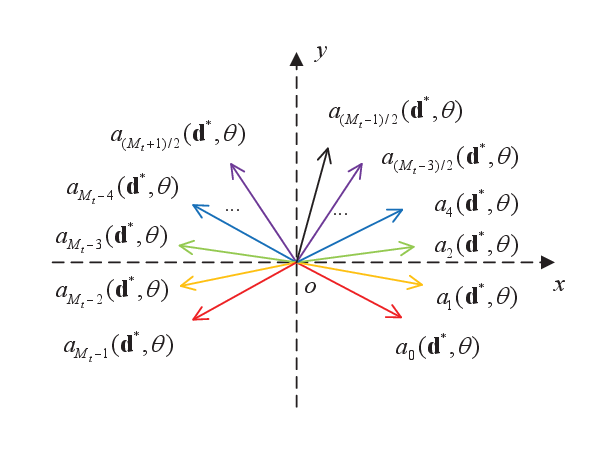}}
		\centering{\small(b)}\label{figure:fig.1020b}	
	\end{subfigure}
	\caption{Illustration of $a_m$ distribution for $m=0,1,\cdots,M_t-1$.}	
	\label{figure:fig.1020}	
\end{figure}

Next, we derive the explicit expressions of $\chi_\mathbf{d}(\theta,\Delta\theta_{a})_y$ and $\chi_{\mathbf{d}^*}(\theta,\Delta\theta_{a})_y$, and compare their values. Since  $d_{t,0}=0$, we can obtain from \eqref{eq:0506_2} that $\psi_0(\mathbf{d},\Delta\theta_a)=\psi_0(\mathbf{d}^*,\Delta\theta_a)$, which means that $a_0(\mathbf{d},\Delta\theta_a)=a_0(\mathbf{d}^*,\Delta\theta_a)$ holds for any $\mathbf{d}$. Let $\beta$ denote the angle between $a_0(\mathbf{d}^*,\Delta\theta_a)$ and the positive x-axis, then the angle between  $a_m(\mathbf{d},\Delta\theta_a),m=0,1,\cdots,M_t-1$ and the positive x-axis can be obtained as $\rho_m(\mathbf{d},\Delta\theta_{a})=\psi_m(\mathbf{d},\Delta\theta_{a})+\beta$. Thus, $\chi_\mathbf{d}(\theta,\Delta\theta_{a})_y$ and $\chi_{\mathbf{d}^*}(\theta,\Delta\theta_{a})_y$ can be expressed as
\begin{equation}\label{eq:1105_1}
	\small
	\chi_{\mathbf{d}}(\theta,\Delta\theta_a)_y=\sum_{m=0}^{M_t-1}\sin(\rho_m(\mathbf{d},\Delta\theta_a)),
\end{equation}
and
\begin{equation}\label{eq:1105_2}
	\small
	\chi_{\mathbf{d}^*}(\theta,\Delta\theta_a)_y=\sum_{m=0}^{M_t-1}\sin(\rho_m(\mathbf{d^*},\Delta\theta_a)),
\end{equation}
respectively.
Based on \eqref{eq:1105_1} and \eqref{eq:1105_2}, we can infer that for any $m\in\mathcal{M}_t$, if $\sin(\rho_m(\mathbf{d},\Delta\theta_a))\geq\sin(\rho_m(\mathbf{d}^*,\Delta\theta_a))$ is satisfied, then $\lvert\chi_{\mathbf{d}}(\theta,\Delta\theta_a)_y\lvert\geq\lvert\chi_{\mathbf{d}^*}(\theta,\Delta\theta_a)_y\lvert$ must hold. To prove this, we resort to the monotonicity of $\sin(\rho_m(\mathbf{d},\Delta\theta_a))$ on $\rho_m(\mathbf{d},\Delta\theta)\in[\psi_{m}^{min}(\Delta\theta_a)+\beta,\psi_{m}^{max}(\Delta\theta_a)+\beta]$. 
Besides, since $a_m(\mathbf{d}^*,\Delta\theta)$ and  $a_{M_t-m-1}(\mathbf{d}^*,\Delta\theta),m=0,1,\cdots,M_t-1$ are symmetric with respect to the y-axis, the analysis for the case of $m\in{\{1,2,\cdots,\frac{M_t}{2}-1\}}$ is similar to that of $m\in{\{\frac{M_t}{2},\frac{M_t}{2}+1,\cdots,{M_t}-1\}}$. Therefore, we will only focus on the former case in the following. 

 According to \eqref{eq:1020_2}, we can see that $a_m(\mathbf{d}^*,\Delta\theta_a), m\in{\{1,2,\cdots,\frac{M_t}{2}-1\}}$ are all located on the right hand side of the y-axis, which means that $\rho_m(\mathbf{d}^*,\Delta\theta_{a})=\psi_m^{min}(\Delta\theta_{a})+\beta<\frac{\pi}{2}$, for all $m\in\{1,2,\cdots,\frac{M_t}{2}-1\}$. If $\psi_m^{max}(\Delta\theta_{a})+\beta<\frac{\pi}{2}$, then $\sin(\rho_m(\mathbf{d},\Delta\theta_a))$ is a monotonically increasing function of $\rho_m(\mathbf{d},\Delta\theta_a)$. Therefore, $\sin(\rho_m(\mathbf{d},\Delta\theta_a))\geq\sin(\rho_m(\mathbf{d}^*,\Delta\theta_a))$ holds for any $\mathbf{d}$. Conversely, if $\psi_m^{max}(\Delta\theta_{a})+\beta\geq\frac{\pi}{2}$, then $\sin(\rho_m(\mathbf{d},\Delta\theta_{a}))$ first increases with $\rho_m(\mathbf{d},\Delta\theta_{a})$ in $\rho_m(\mathbf{d},\Delta\theta_{a})\in[\psi_m^{min}(\Delta\theta_{a})+\beta,\frac{\pi}{2}]$, and then decreases with $\rho_m(\mathbf{d},\Delta\theta_{a})$ in $\rho_m(\mathbf{d},\Delta\theta_{a})\in[\frac{\pi}{2},\psi_m^{max}(\Delta\theta_{a})+\beta]$. Therefore, to determine the minimum value of $\sin(\rho_m(\mathbf{d},\Delta\theta_a))$ in $\rho_m(\mathbf{d},\Delta\theta_a)\in[\psi_m^{min}(\Delta\theta_{a})+\beta,\psi_m^{max}(\Delta\theta_{a})+\beta]$, the values of $\sin(\psi_m^{min}(\Delta\theta_{a})+\beta)$ and $\sin(\psi_m^{max}(\Delta\theta_{a})+\beta)$ are required to be compared. 
However, directly perform the above comparison using algebraic methods  is very difficult, therefore, we propose to use $\sin(\psi_{M_t-m-1}^{max}(\Delta\theta_a)+\beta)$ as an intermediate variable to simplify the comparison process. Specifically, based on \eqref{eq:1021_7}, $\sin(\psi_m^{min}(\Delta\theta_a)+\beta)=\sin(\psi_{M_t-m-1}^{max}(\Delta\theta_a)+\beta)$ must be satisfied. Meanwhile, it can be observed that $\psi_{M_t-m-1}^{max}>\psi_{m}^{max}\geq\frac{\pi}{2}$ when $m\in{\{1,2,\cdots,\frac{M_t}{2}-1\}}$, which implies that $\sin(\psi_{m}^{max}+\beta)>\sin(\psi_{m}^{min}+\beta)$ holds. Thus, when $\rho_m(\mathbf{d},\Delta\theta_a)\in[\psi_m^{min}(\Delta\theta_{a})+\beta,\psi_m^{max}(\Delta\theta_{a})+\beta]$, the minimum value of $\sin(\rho_m(\mathbf{d},\Delta\theta_a))$ is $\sin(\psi_{m}^{min}(\Delta\theta_a)+\beta)$, and for any $\mathbf{d}$, $\sin(\rho_m(\mathbf{d},\Delta\theta_a)+\beta)\geq\sin(\rho_m(\mathbf{d}^*,\Delta\theta_a))$ holds for all $m\in\{0,1,\cdots,\frac{M_t}{2}-1\}$. Similarly, we have that for any $\mathbf{d}$, $\sin(\rho_m(\mathbf{d},\Delta\theta_a)+\beta)\geq\sin(\rho_m(\mathbf{d}^*,\Delta\theta_a))$ holds for all $m\in\{\frac{M_t}{2},\frac{M_t}{2}+1,\cdots,M_t-1\}$. Therefore, when $M_t$ is even, $\lvert\chi_\mathbf{d}(\theta,\Delta\theta_{a})\lvert\geq\lvert\chi_{\mathbf{d}^*}(\theta,\Delta\theta_{a})\lvert$ holds for any $\mathbf{d}$ that satisfies the condition $d_{t,i}\geq\frac{\lambda}{2},i\in\mathcal{M}_t$ and $\sum_{i=0}^{M_t-1}d_{t,i}\leq L$.

For the case when $M_t$ is odd, a similar approach can be applied to show that $\mathbf{d}^*$ achieves the minimum main lobe width, however, some modifications are required as will be explained in the following.  First, we can see that, in this case, \eqref{eq:1020_2} becomes
\begin{equation}\label{eq:1106_1}
	\small
	\begin{aligned}
		&\psi_m(\mathbf{d}^*,\Delta\theta_{a})\\&=\left\{\begin{aligned}
			&\psi_m^{min}(\Delta\theta_{a}),\quad m=0,1,\cdots,\frac{M_t-1}{2}\\
			&\psi_m^{max}(\Delta\theta_{a}),\quad m=\frac{M_t+1}{2},\frac{M_t+3}{2},\cdots,M_t-1\\
		\end{aligned}
		\right
		.,
	\end{aligned}
\end{equation}
from which we can further obtain 
\begin{equation}\label{eq:1106_2}
	\small
	\begin{aligned}
		&\psi_m(\mathbf{d^*},\Delta\theta_{a})+\psi_{M_t-m}(\mathbf{d^*},\Delta\theta_{a})\\&=\psi_{(M_t-3)/2}(\mathbf{d^*},\Delta\theta_{a})+\psi_{(M_t+1)/2}(\mathbf{d^*},\Delta\theta_{a}),
	\end{aligned}
\end{equation}
when $m\in\mathcal{M}_t,m\neq\frac{M_t-1}{2}$, and
\begin{equation}\label{eq:1106_3}
	\small
	\begin{aligned}
	&\psi_{(M_t-1)/2}^{min}(\Delta\theta_{a})+\psi_{(M_t-1)/2}^{max}(\Delta\theta_{a})\\&=\psi_{(M_t-3)/2}(\mathbf{d^*},\Delta\theta_{a})+\psi_{(M_t+1)/2}(\mathbf{d^*},\Delta\theta_{a}).
\end{aligned}
\end{equation}
The equality in \eqref{eq:1106_2} means that  $a_m(\mathbf{d}^*,\Delta\theta_{a})$ and  $a_{M_t-m-1}(\mathbf{d}^*,\Delta\theta_{a}),m\in\mathcal{A}$ are symmetric with respect to
\begin{equation}\label{eq:1107_2}
	\small
	y_2\triangleq a_{(M_t-3)/2}(\mathbf{d}^*,\Delta\theta_{a})+a_{(M_t+1)/2}(\mathbf{d},\Delta\theta_{a}),
\end{equation}
where $\mathcal{A}\triangleq\{m|m=0,1,\cdots,M_t-1,m\neq(M_t-1)/2\}$. In addition, according to \eqref{eq:1106_3}, $e^{j(\psi_{(M_t-1)/2}^{min}(\Delta\theta_{a})+\beta)}$ and $e^{j(\psi_{(M_t-1)/2}^{max}(\Delta\theta_{a})+\beta)}$ are also symmetric with respect to $y_2$. Therefore, we utilize the intersection point of $a_m(\mathbf{d}^*,\Delta\theta_{a}),m=0,1,\cdots,M_t-1$ as the origin and the direction of $y_2$ as the positive y-axis to construct a CCS in this case, as illustrated in Fig. \ref{figure:fig.1020} (b).

It is observed from Fig. \ref{figure:fig.1020} (b) that $a_m(\mathbf{d}^*,\Delta\theta_{a})$ and $a_{M_t-m-1}(\mathbf{d}^*,\Delta\theta_{a}),m\in\mathcal{A}$ are symmetric with respect to $y_2$, which is similar to the case when $M_t$ is even. Therefore, the following argument can be similarly proved based on the proof for the even $M_t$ case: when $\lvert\chi_{\mathbf{d}}(\theta,\Delta\theta)\lvert$ achieves its minimum value, 
\begin{equation}\label{eq:1107_1}
	\small
	\begin{aligned}
		&a_m(\mathbf{d},\Delta\theta_{a})=a_m(\mathbf{d}^*,\Delta\theta_{a}),\\&m=0,1,\cdots,M_t-1,m\neq\frac{M_t-1}{2}
	\end{aligned}
\end{equation}
will be satisfied. Then, we only need to prove that for any $\mathbf{d}$ that satisfies the condition \eqref{eq:1107_1}, $\lvert\chi_{\mathbf{d}}(\theta,\Delta\theta_a)\lvert\geq\lvert\chi_{\mathbf{d}^*}(\theta,\Delta\theta_a)\lvert$ will be satisfied.

To derive the expressions of $\lvert\chi_{\mathbf{d}}(\theta,\Delta\theta_a)\lvert$ and $\lvert\chi_{\mathbf{d}^*}(\theta,\Delta\theta_a)\lvert$, we first analyze $\chi_{\mathbf{d}}(\theta,\Delta\theta_a)_x$ and $\chi_{\mathbf{d}^*}(\theta,\Delta\theta_a)_x$, which can be respectively written as
\begin{equation}\label{eq:1106_4}
	\small
	\chi_{\mathbf{d}}(\theta,\Delta\theta_a)_x = \sum_{m=0}^{M_t-1} \cos(\rho_m(\mathbf{d},\Delta\theta_a)),
\end{equation}
and
\begin{equation}\label{eq:1106_5}
	\small
	\chi_{\mathbf{d}^*}(\theta,\Delta\theta_a)_x = \sum_{m=0}^{M_t-1} \cos(\rho_m(\mathbf{d}^*,\Delta\theta_a)).
\end{equation}
Since $a_m(\mathbf{d}^*,\Delta\theta_{a})$ and $a_{M_t - m - 1}(\mathbf{d}^*,\Delta\theta_{a})$ (for $m \neq \frac{M_t - 1}{2}$) are symmetric with respect to $y_2$, we have
\begin{equation}\label{eq:1106_6}
	\small
	\cos(\rho_m(\mathbf{d}^*,\Delta\theta_{a})) + \cos(\rho_{M_t - m - 1}(\mathbf{d}^*,\Delta\theta_{a})) = 0,
\end{equation}
for all $m \in \{0, 1, \dots, M_t - 1\}, m \neq \frac{M_t - 1}{2}$.
Thus, when $\mathbf{d}$ satisfies \eqref{eq:1107_1}, \eqref{eq:1106_4} and \eqref{eq:1106_5} can be simplified to $	\chi_{\mathbf{d}}(\theta,\Delta\theta_a)_x = \cos(\rho_{(M_t - 1)/2}(\mathbf{d},\Delta\theta_a))$ and $\chi_{\mathbf{d}^*}(\theta,\Delta\theta_a)_x = \cos(\rho_{(M_t - 1)/2}(\mathbf{d}^*,\Delta\theta_a))$.
Then based on the expressions of $\chi_{\mathbf{d}}(\theta,\Delta\theta_a)_y$ and $\chi_{\mathbf{d}^*}(\theta,\Delta\theta_a)_y$ given in \eqref{eq:1105_1} and \eqref{eq:1105_2},  $\lvert\chi_{\mathbf{d}}(\theta,\Delta\theta_a)\lvert$ and $\lvert\chi_{\mathbf{d}^*}(\theta,\Delta\theta_a)\lvert$ can be expressed as
\begin{equation}\label{eq:1106_9}
	\small
	\begin{aligned}
		&\lvert\chi_{\mathbf{d}}(\theta,\Delta\theta_a)\lvert\\&=1+2\bigg(\!\sum_{m\in\mathcal{A}}\!\!\!\sin(\rho_m(\mathbf{d},\Delta\theta_a)\bigg)\sin(\rho_{\frac{M_t-1}{2}}(\mathbf{d},\Delta\theta_a))\\&+\bigg(\!\sum_{m\in\mathcal{A}}\!\!\!\sin(\rho_m(\mathbf{d},\Delta\theta_a)\bigg)^2
	\end{aligned}
\end{equation}
and
\begin{equation}\label{eq:1106_10}
	\small
	\begin{aligned}
		&\lvert\chi_{\mathbf{d}^*}(\theta,\Delta\theta_a)\lvert\\&=1+2\bigg(\!\sum_{m\in\mathcal{A}}\!\!\!\sin(\rho_m(\mathbf{d}^*,\Delta\theta_a)\bigg)\sin(\rho_{\frac{M_t-1}{2}}(\mathbf{d}^*,\Delta\theta_a))\\&+\bigg(\sum_{m\in\mathcal{A}}\!\!\!\sin(\rho_m(\mathbf{d}^*,\Delta\theta_a)\bigg)^2.
	\end{aligned}
\end{equation}
When the antenna distribution $\mathbf{d}$ satisfies \eqref{eq:1107_1}, \eqref{eq:1106_9} can be rewritten as
\begin{equation}\label{eq:1106_11}
	\small
	\begin{aligned}
		&\lvert\chi_{\mathbf{d}}(\theta,\Delta\theta_a)\lvert\\&=1+2\bigg(\!\sum_{m\in\mathcal{A}}\!\!\sin(\rho_m(\mathbf{d}^*,\Delta\theta_a)\bigg)\sin(\rho_{\frac{M_t-1}{2}}(\mathbf{d},\Delta\theta_a))\\&+\bigg(\!\sum_{m\in\mathcal{A}}\!\!\sin(\rho_m(\mathbf{d}^*,\Delta\theta_a)\bigg)^2.
	\end{aligned}
\end{equation}
From \eqref{eq:1106_10} and \eqref{eq:1106_11}, it follows that to compare $\lvert\chi_{\mathbf{d}}(\theta,\Delta\theta_a)\lvert$ and $\lvert\chi_{\mathbf{d}^*}(\theta,\Delta\theta_a)\lvert$, we need to evaluate the values of $\sin(\rho_{(M_t-1)/2}(\mathbf{d},\Delta\theta_a))$ and $\sin(\rho_{(M_t-1)/2}(\mathbf{d}^*,\Delta\theta_a))$. Similar to the case when $M_t$ is even, we can obtain that for any antenna distribution $\mathbf{d}$, $\sin(\rho_{(M_t-1)/2}(\mathbf{d},\Delta\theta_a)) \geq \sin(\rho_{(M_t-1)/2}(\mathbf{d}^*,\Delta\theta_a))$ holds.
Thus, we have $\lvert\chi_{\mathbf{d}}(\theta,\Delta\theta_a)\lvert \geq \lvert\chi_{\mathbf{d}^*}(\theta,\Delta\theta_a)\lvert$, which implies that $\mathbf{d}^*$ minimizes the main lobe width when $M_t$ is odd.

Based on the above analysis, we conclude that the proposed antenna distribution $\mathbf{d}^*$ is optimal for minimizing the main lobe width of the ambiguity function. The first null point of the ambiguity function $\Delta\theta_{a1}$ can then be obtained by setting $\frac{\partial\lvert\chi_{\mathbf{d}^*}(0,0,\theta,\theta+\Delta\theta_{a1})\lvert^2}{\partial\Delta\theta_{a1}}=0$, and the detailed proof is omitted here for brevity. This completes the proof.
\end{proof}
\vspace{-0.1cm}
\begin{figure}[t]
	\centering
	\includegraphics[width=0.45\textwidth]{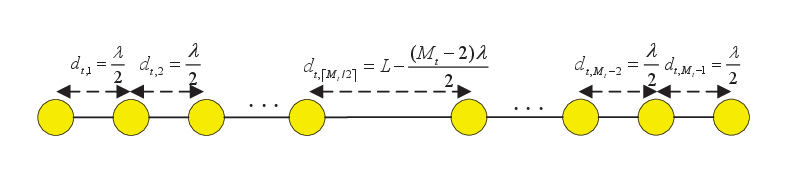}
	\caption{Illustration of the optimal MAs’ positions for minimum main lobe width.}	
	\label{figure:fig.01017}	
	\vspace{-0.4cm}
\end{figure}
Theorem \ref{thm1} demonstrates that to minimize the main lobe width, the MAs should be partitioned into two groups as shown in Fig. \ref{figure:fig.01017}, which is defined as minimum main lobe width distribution (MMLWD).
Also, several key observations can be obtained. First, we can see that $B_{min}$ decreases as $L/\lambda$ increases and as $M_t$ and $\lvert\theta\lvert$ decrease. With the increasing of $L/\lambda$, the performance gain growth will gradually decrease, which means that employing an antenna array with excessive large dimension may not be very cost-efficient. Then, based on the relationship between $B_{min}$ and $M_t$, it can be observed that a smaller antenna number could actually lead to narrower main lobe width in practical applications. But the side lobe performance in this case cannot be guaranteed. Therefore, in practice, the value of $M_t$ should be properly chosen to achieve a good balance between main lobe and side lobe performance. Finally, from \eqref{eq:21}, we can infer that the best angular resolution of the considered system is achieved at $\theta=0$. 
 
\subsection{Doppler and delay Domains} \label{section3-B}
 Different from the angular domain, the relationship between the ambiguity function and the antenna distribution in the Doppler and delay domains are more complex since it is very difficult to simplify the impact of FH code on the ambiguity function, as can be observed in \eqref{eq:15}. This implies that it is difficult to derive the specific antenna distribution that achieves the minimum main lobe width in these two domains. However, by analyzing \eqref{eq:15}, we discover that lower bounds exist for the ambiguity function, which correspond to the performance of the considered MA-enabled MIMO radar system under ideal conditions, and they can be used as a theoretical underpinning for enhanced radar performance through antenna distribution optimization. The main results are given in the following theorems.
 \newtheorem{theorem1}{\bf Theorem}
 \begin{theorem}\label{thm2}
 For a given $v$, assuming that $L$ is sufficiently  large, and the target’s range and direction angle are accurately estimated, i.e., $\tau=0$ and $\theta=\theta'$, a lower bound of the ambiguity function in the Doppler domain, denoted by $\chi^*(v)$ $($i.e., $\chi^*(v)\leq\lvert\chi(0,v,\theta,\theta)\lvert)$ is given by
 \begin{equation}\label{eq:0611_1}
 	\small
 	\chi^*(v)=\left\{\begin{aligned}
 	&\lvert M_t{\rm sinc}(v\Delta_t)\lvert-\Xi^*(v),\quad \lvert M_t{\rm sinc}(v\Delta_t)\lvert\geq \Xi^*(v)\\
 	&0,\quad  \mathrm{otherwise}
 	\end{aligned}
 	\right.
 	,
 \end{equation}
where
\begin{equation}\label{eq:0611_2}
	\small
	\Xi^*(v)=\sum_{\substack{m,m'=0\\m\neq m'}}^{M_t-1}\sum_{q=0}^{Q-1}\lvert{\rm sinc}(v\Delta t-c_{m,q}+c_{m',q})\lvert.
\end{equation}
\end{theorem}
 \begin{proof}
 	Please refer to Appendix \ref{A.A_1}.
\vspace{-0.1cm}
 \end{proof}
\newtheorem{theorem2}{theorem}
\begin{theorem}\label{thm3}
 	For a given $\tau$, assuming that $L$ is sufficiently large, and  the target’s velocity and direction angle are accurately estimated, i.e., $v=0$ and $\theta=\theta'$, the lower bound of the ambiguity function in the delay domain $\chi^d(\tau)$ $($i.e., $\chi^d(\tau)\leq\lvert\chi(\tau,0,\theta,\theta)\lvert)$ is given by
 	\begin{equation}\label{eq:0613_1}
 		\small
 			\begin{aligned}
 					\chi^d(\tau)=\left\{\begin{aligned}
 						&\lvert\Upsilon^d(\tau)\lvert-\Xi^d(\tau),\quad  \lvert\Upsilon^d(\tau)\lvert\geq\Xi^d(\tau)\\
 						&0,\quad\mathrm{otherwise}
 						\end{aligned}
 					\right.
 				,
 				\end{aligned}
 		\end{equation}
 where
 \begin{equation}\label{eq:0613_2}
 	\small
 	\begin{aligned}
 			\Upsilon^d(\tau)\triangleq\sum_{m\!=0}^{M_t\!-\!1}\sum_{q,q'\!=0}^{Q-1}&\chi^{r}(\tau\!-\!(q'\!-\!q){\Delta}_t,(c_{m,q}-c_{m,q'})\Delta_f)\\&e^{-j2{\pi\Delta}_fc_{m,q'}\tau}
 		\end{aligned}
 \end{equation}
 and
 \begin{equation}\label{eq:0613_3}
 	\small
 	\begin{aligned}
 			&\Xi^d(\tau)\triangleq\!\!\!\!\!\sum_{\substack{m,m'=0\\m\neq m'}}^{M_t\!-\!1}\!\sum_{q,q'\!=0}^{Q-1}\!\lvert\chi^{r}(\tau\!-\!(q'\!-\!q){\Delta}_t,\!(c_{m \!,q}-c_{m'\!\!,q'})\Delta_f)\lvert.
 		\end{aligned}
 \end{equation}
 \end{theorem}
 \begin{proof}
 Please refer to Appendix \ref{A.A_2}.
\end{proof}
 
Theorems \ref{thm2} and \ref{thm3} are very useful, as they establish lower bounds of the ambiguity function in the Doppler and delay domains, respectively.  In Theorem \ref{thm2}, \eqref{eq:0611_1} implies that  an increase in the signal bandwidth will result  in a larger  main lobe width of the ambiguity function in the Doppler domain. Besides, in Theorem 3, based on the definition of $\chi^r(\tau,v)$ given in (9), it can be inferred that as  $\Delta_f$ increases, the main lobe width of $\chi^r(\tau-(q'-q)\Delta_f,(c_{m,q}-c_{m',q'})\Delta_f)$ with respect to $\tau$ decreases. Consequently, a larger $\Delta_f$ results in a reduction in the main lobe widths of $\lvert\Upsilon^d(\tau)\lvert$ and $\Xi^d(\tau)$, thereby leading to a further contraction of the main lobe width of $\chi^d(\tau)$. This implies that increasing the signal bandwidth reduces the main lobe width of the ambiguity function in the delay domain. Furthermore, Theorems \ref{thm2} and \ref{thm3}
 reveal that, under the assumption of a sufficiently large $L$, both $\chi^*(v)$ and $\chi^d(\tau)$ become independent of the antenna size $L$, which is mainly due to the periodic nature of the antenna distribution’s influence on the ambiguity function, as captured by $e^{j2\pi J_{m,m',\theta}(\mathbf{d})}$. This result indicates that performance of the ambiguity function  in these two domains will not improve as $L$ increases if $L$ is very large, which is different from the results observed in the angular domain. Furthermore, we observe that the performance gain of the MAs in both domains is minimal when $\theta = 0$. This is because, when $\theta = 0$, we can observe that $e^{j2\pi J_{m,m',\theta}(\mathbf{d})} = 1$, which implies that $e^{j2\pi J_{m,m',\theta}(\mathbf{d})}$ becomes a constant that is independent of the antenna distribution $\mathbf{d}$. 
Finally, the derived lower bounds can serve as performance benchmarks for the optimization of antenna distribution, as discussed in the next section.

\section{Problem formulation and proposed Algorithm}
 In this section, we formulate an optimization problem of the antenna distribution $\mathbf{d}$ to achieve a more balanced performance between the main lobe and side lobes. Then, a low-complexity GPRM-based algorithm is proposed to address this problem.
\subsection{Problem Formulation}
First,  we illustrate in Fig. \ref{figure:fig.0428} a numerical example of the MMLWD given in \eqref{eq:0406_2} based on the simulation setup in Section V, and we set $L=9\lambda$, $\theta=0$ and $M_t=8$. As can be observed, the MMLWD is able to effectively reduce  the main lobe width as compared to an equally spaced antenna array with half wavelength spacing; however, its side lobe peak is very high and close to the main lobe peak, which will deteriorate the angular resolution and anti-interference capabilities of the proposed system. Motivated by this, in this subsection, we formulate an optimization problem to strike a balance between main lobe and side lobe performance in the angular/Doppler/delay domains.  

\begin{figure}[t]
	\centering
	\includegraphics[width=0.4\textwidth]{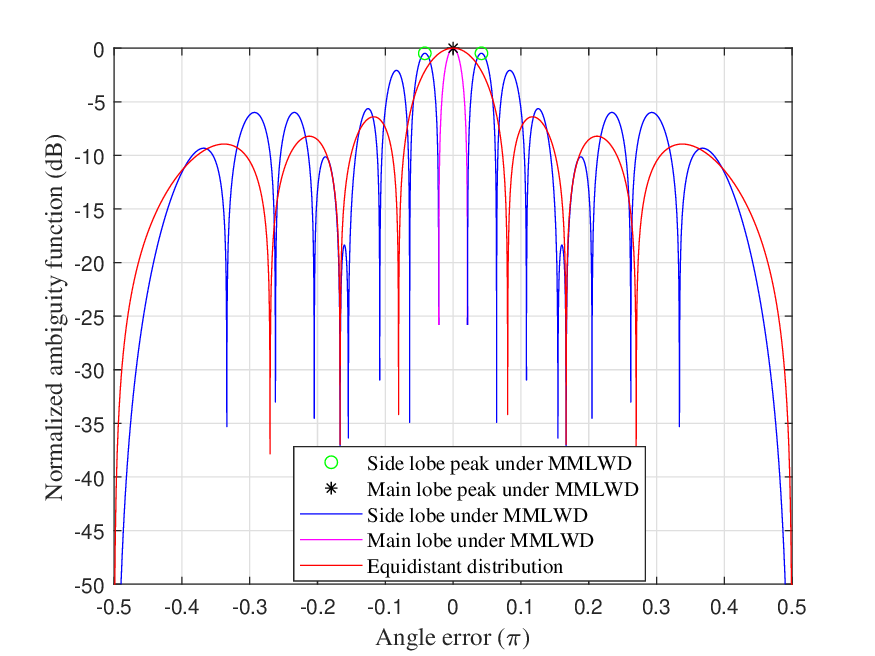}
	\caption{Ambiguity function in the angular domain.}	
	\label{figure:fig.0428}	
	\vspace{-0.4cm}
\end{figure}

In the angular domain, the main lobe peak of the ambiguity function is achieved when the parameters of the matched filter and the target are well matched, where we have $\chi(0,0,\theta,\theta)=M_t$ by using \eqref{eq:20} and assuming $\tau=0$, $v=0$. Therefore, we need to eliminate the peaks in $\lvert\chi(0,0,\theta,\theta')\lvert$, which are not in the line $\{0,0,\theta,\theta\}$. This can be done by imposing a cost function which puts penalties on these peak values\cite{chenMIMORadarAmbiguity2008} and forces the energy of the ambiguity function $\lvert\chi(0,0,\theta,\theta')\lvert$ to be evenly spread. Similar to the angular domain, it is necessary to mitigate the peaks in $\lvert\chi(\tau,0,\theta,\theta)\lvert$ and $\lvert\chi(0,v,\theta,\theta)\lvert$, which correspond to the Doppler and delay domains, respectively. Generally, the following optimization  problem can be formulated:
\begin{align}
	\small
		\min_{\mathbf{d}}  &\quad f_j(\mathbf{d}) \quad \label{pro:1}\\
		\text{s.t.}\ &\quad\sum_{i=1}^{M_t-1}d_{t,i}\leq L, \tag{\ref{pro:1}{a}}\\
		&\quad d_{t,i}\geq \frac{\lambda}{2},i=1,2,\cdots,M_t-1,\tag{\ref{pro:1}{b}}
\end{align}
where $j=1,2,3$ and
\begin{equation}\label{eq:0716_1}
	\small
	f_1(\mathbf{d})=\int_{-\frac{\pi}{2}}^{+\frac{\pi}{2}}\int_{-\frac{\pi}{2}}^{+\frac{\pi}{2}}\lvert\chi(0,0,\theta,\theta')\lvert^2 d\theta' d\theta,
\end{equation} 
 \begin{equation}\label{eq:0716_3}
 	\small
 	f_2(\mathbf{d})=\int_{-\frac{\pi}{2}}^{+\frac{\pi}{2}}\int_{-f_{max}}^{+f_{max}}\lvert\chi(0,v,\theta,\theta)\lvert^2 dvd\theta,
 \end{equation}
 \begin{equation}\label{eq:0716_2}
 	\small
 	f_3(\mathbf{d})=\int_{-\frac{\pi}{2}}^{+\frac{\pi}{2}}\int_{-Q\Delta t}^{+Q\Delta t}\lvert\chi(\tau,0,\theta,\theta)\lvert^2 d\tau d\theta,
 \end{equation}
represent the objective functions in the angular, Doppler and delay domains, respectively. In practice, however, these objectives may be required to be optimized simultaneously. Thus, we propose to use the weighted sum of $f_j(\mathbf{d}),j=1,2,3$ as the objective function and problem \eqref{pro:1} can be transformed into 
	\begin{align}
		\small
		\min_{\mathbf{d}}  &\quad f_a(\mathbf{d})\triangleq\sum_{j=1}^{3}\alpha_jf_j(\mathbf{d})\label{pro:2}\\
		\text{s.t.}\ &\quad\sum_{i=1}^{M_t-1}d_{t,i}\leq L,\tag{\ref{pro:2}{a}}\\
		&\quad d_{t,i}\geq \frac{\lambda}{2},i=1,2,\cdots,M_t-1\tag{\ref{pro:2}{b}}
	\end{align}
where $\alpha_j$ is the weight coefficient corresponding to the $j$-th objective function and  $\alpha$'s are enforced to satisfy  $\sum_{j=1}^{3}\alpha_j=1$.
\subsection{Proposed Algorithm}
Problem \eqref{pro:2} is highly non-convex and challenging to solve due to the integral operations in its objective function. Moreover, $\lvert\chi(0,0,\theta,\theta')\lvert^2$, $\lvert\chi(\tau,0,\theta,\theta')\lvert^2$, and $\lvert\chi(0,v,\theta,\theta')\lvert^2$ are generally exponential functions with respect to the antenna positions in $\mathbf{d}$, which are also very difficult to handle. In existing works, meta-heuristic algorithms, such as the simulated annealing and genetic algorithm (GA), are usually employed to solve such kind of problems , as in \cite{chenMIMORadarAmbiguity2008} and \cite{eedaraDualFunctionFrequencyHoppingMIMO2022}. While these methods yield satisfactory performance, their high computational complexity renders them unsuitable for real-time applications in the considered system, where dynamic adjustments of the antenna positions are crucial for attaining different resolutions in the angle, Doppler, and delay domains. Thus, we propose a low-complexity algorithm based on the RGPM to effectively solve problem \eqref{pro:1}, and we choose GA as the baseline.

To proceed, we first transform the complex integral-based objective function into a more manageable form based on summation, facilitating the subsequent gradient calculation process.  Specifically, we apply the Riemann sum method to approximate $f_j(\mathbf{d})$'s, and the corresponding approximated functions can be expressed as follows:
\begin{equation}\label{eq:49}
	\bar{f}_1(\mathbf{d})=\sum_{i_1=0}^{n_1}\sum_{i_2=0}^{n_1}\lvert\chi(0,0,\theta_{i_1},\theta'_{i_2})\lvert^2 \Delta\theta' \Delta\theta,
\end{equation}
\begin{equation}\label{0919_2}
	\bar{f}_2(\mathbf{d})=\sum_{i_1=0}^{n_1}\sum_{i_3=0}^{n_2}\lvert\chi(0,v_{i_3},\theta_{i_1},\theta_{i_1})\lvert^2 \Delta v\Delta\theta,
\end{equation}
\begin{equation}\label{0919_1}
	\bar{f}_3(\mathbf{d})=\sum_{i_1=0}^{n_1}\sum_{i_4=0}^{n_3}\lvert\chi(\tau_{i_4},0,\theta_{i_1},\theta_{i_1})\lvert^2 \Delta \tau\Delta\theta,
\end{equation}
where $\lvert\chi(\tau,v,\theta,\theta')\lvert^2$ is partitioned into $n_1$, $n_2$, and $n_3$ equal-length segments along the angle, Doppler, and delay dimensions, respectively. For the angular domain, we define  $\Delta\theta\triangleq\frac{\pi}{n_1}$, $\Delta\theta'\triangleq\frac{\pi}{n_1}$, $\theta_{i_1}\triangleq-\frac{\pi}{2}+ i_1 \Delta\theta$ and $\theta'_{i_2}\triangleq-\frac{\pi}{2}+ i_2 \Delta\theta'$. Similarly, we define $\Delta v \triangleq \frac{2f_{max}}{n_2}$, $v_{i_3}\triangleq-f_{max}+i_3\Delta v$, $\Delta \tau \triangleq \frac{2Q\Delta t}{n_3}$, and $\tau_{i_4}\triangleq-Q\Delta \tau+i_4 \Delta \tau$ for the Doppler and delay domains. According to the Nyquist Sampling Theorem \cite{c2009radar}, to reduce the impact of the above approximation on the ambiguity function, the sampling intervals must satisfy $\Delta\theta\leq\arcsin(\frac{\lambda}{L})$, $\Delta v\leq\frac{1}{T_w}$, and $\Delta\tau \leq\frac{1}{2K\Delta_f}$. Consequently, $n_1$, $n_2$ and $n_3$ should satisfy $n_1\geq\frac{2\pi}{B_{min}}$, $n_2\geq4f_{max}T_w$, and $n_3\geq4Q\Delta tK\Delta_f$.
 In addition, to reduce the complexity of calculating $\bar{f}(\mathbf{d}) \triangleq\sum_{j=1}^{3}\alpha_j\bar{f}_j(\mathbf{d})$, due to the dependence of $\lvert\chi(\tau,v,\theta,\theta')\lvert^2$ on $\mathbf{d}$, we reformulate $\lvert\chi(\tau,v,\theta,\theta')\lvert^2$ using Euler's formula, as detailed in Appendix \ref{A.A}.

Next, based on the above transformations, we address problem \eqref{pro:1} by iteratively minimizing $\bar{f}(\mathbf{d})$ along its negative gradient direction, i.e., $-\nabla\bar{f}(\mathbf{d})=-[\frac{\partial\bar{f}(\mathbf{d})}{\partial d_{t,0}},\frac{\partial\bar{f}(\mathbf{d})}{\partial d_{t,1}},\cdots,\frac{\partial\bar{f}(\mathbf{d})}{\partial d_{t,M_t-1}}]^T$, where
\begin{equation}\label{eq:0520_2}
	\frac{\partial\bar{f}(\mathbf{d})}{\partial d_{t,x}}=\alpha_1\frac{\partial\bar{f}_1(\mathbf{d})}{\partial d_{t,x}}+\alpha_2\frac{\partial\bar{f}_2(\mathbf{d})}{\partial d_{t,x}}+\alpha_3\frac{\partial\bar{f}_3(\mathbf{d})}{\partial d_{t,x}}.
\end{equation}
To obtain $\frac{\partial\bar{f}(\mathbf{d})}{\partial d_{t,x}}$, we should compute the derivatives $\frac{\partial\lvert\chi(0,0,\theta_{i_1},\theta'_{i_2})\lvert^2}{\partial d_{t,x}}$, $\frac{\partial\lvert\chi(\tau_{i_3},0,\theta_{i_1},\theta_{i_1})\lvert^2 }{\partial d_{t,x}}$, and $\frac{\partial\lvert\chi(0,v_{i_4},\theta_{i_1},\theta_{i_1})\lvert^2 }{\partial d_{t,x}}$, and the detailed derivations are provided in Appendix \ref{A.B}.

After  each gradient descent iteration, $\mathbf{d}^k$ ($k$ denotes the iteration index) must remain within this region. As the feasible region of problem \eqref{pro:2} is determined by (\ref{pro:2}a) and (\ref{pro:2}b), we introduce a projection matrix $\mathbf{P}^k$ based on RGPM, which is given by 
\begin{equation}\label{eq:42}
	\mathbf{P}^k = \mathbf{I}-(\mathbf{M}^k)^T(\mathbf{M}^k(\mathbf{M}^k)^T)^{-1}\mathbf{M}^k,
\end{equation}
where $\mathbf{M}^k$ represents the active boundary constraints of $\mathbf{d}^k$. Specifically, if $\mathbf{d}^k$ lies strictly within the feasible region, then $\mathbf{M}^k$ is an empty matrix, which indicates that all inequality constraints are strictly satisfied, with none of them holding as an equality. To obtain $\mathbf{M}^k$ when $\mathbf{d}^k$ is on the boundary, we first equivalently rewrite  the constraints of \eqref{pro:1} as $\mathbf{A}\mathbf{d}^k \succeq \mathbf{b}$, where $\mathbf{A} = \begin{bmatrix} \mathbf{I}; -\mathbf{1}_{1\times (M_t-1)} \end{bmatrix}$ and $\mathbf{b} = \begin{bmatrix} \frac{\lambda}{2} \mathbf{1}_{(M_t-1)\times 1}; -L \end{bmatrix}$. If $\mathbf{d}^k$ lies on the boundary of the feasible region, we can decompose $\mathbf{A}$ and $\mathbf{b}$ into two parts, known as the  active and inactive parts, i.e., $\mathbf{A} = \begin{bmatrix} \mathbf{A}_1; \mathbf{A}_2 \end{bmatrix}$ and $\mathbf{b} = \begin{bmatrix} \mathbf{b}_1; \mathbf{b}_2 \end{bmatrix}$, where $\mathbf{A}_1\mathbf{d}^k = \mathbf{b}_1$ defines the active constraints (i.e., those satisfied with equality), and $\mathbf{A}_2\mathbf{d}^k \succ \mathbf{b}_2$ defines the inactive constraints. Consequently, we set $\mathbf{M}^k = \mathbf{A}_1$ to represent the active boundary constraints in the current iteration.

Once $\mathbf{M}^k$ is determined, the next step is to compute the step size $\omega^k$, which plays a crucial role in accelerating the proposed algorithm. To ensure efficient convergence, we determine the step size $\omega^k$ in each iteration by solving the following optimization problem:
\begin{align}
	\small
	\min_{\omega^k}  &\quad \bar{f}_a(\mathbf{d}^k-\omega^k\mathbf{P}^k\nabla \bar{f_a}(\mathbf{d}^k)) \label{pro:3}\\
	\text{s.t.}\ &\quad\mathbf{A}_2(\mathbf{d}^k-\omega^k\mathbf{P}^k\nabla \bar{f_a}(\mathbf{d}^k))\succeq 0,\tag{\ref{pro:3}{a}} \label{Problema}\\
	&\quad\omega^k\geq 0\tag{\ref{pro:3}{b}} \label{Problemb}.
\end{align}
To solve this problem, we employ the Armijo rule \cite{Armijo1966}, which is effective in adaptively selecting an appropriate step size $\omega^k$ that satisfies the convergence criteria. 
After obtaining an appropriate value of  $\omega^k$, we update $\mathbf{d}$ in the direction of the negative projected gradient, i.e.,
\begin{equation}\label{eq:37}
	\small
	\mathbf{d}^{k+1}=\mathbf{d}^{k}-\omega^k\mathbf{P}^k\nabla \bar{f_a}(\mathbf{d}^k).
\end{equation}
Note that the proposed algorithm  converges under the following conditions: (1) when $\mathbf{d}^k$ is within the feasible region, then $\|\mathbf{P}\nabla\bar{f_a}(\mathbf{d}^k)\| \!= 0$ should be satisfied; (2) when $\mathbf{d}^k$ is on the boundary, then $\|\mathbf{P}\nabla\bar{f_a}(\mathbf{d}^k)\|\!=\!0$ and  $(\mathbf{M}\mathbf{M}^T)^{-\!1}\mathbf{M}\nabla\bar{f_a}(\mathbf{d}^k)\!\geq \mathbf{0}$ should be satisfied.

\renewcommand\baselinestretch{\linspreadalgr}\selectfont 
\begin{algorithm}[t]
	\small
	\SetAlgoLined 
	\caption{\small Proposed RGPM-based Algorithm}
	\textbf{Input:} Maximum iteration number $K$, threshold $T$, and step size $\omega$. \textbf{Initialize}: $\mathbf{d}^0$, $k=0$\;
	\While{$k<K$}{
		$k \gets k+1$\;
		Obtain $\mathbf{M}^k$ based on $\mathbf{A}\mathbf{d}^k\succeq \mathbf{b}$\;
		Calculate $\mathbf{P}^k$ according to \eqref{eq:42}\;
		\eIf{$\|\mathbf{P}^k\nabla\bar{f_a}(\mathbf{d})\|<T$}{\eIf{$\mathbf{M}^k$ is an empty matrix}{\textbf{Break;}}{ $\mathbf{u}^k=(\mathbf{M}^k(\mathbf{M}^k)^T)^{-1}\mathbf{M}^k\nabla\bar{f_a}(\mathbf{d}^{k})$, obtain the minimum element in $\mathbf{u}^k$ via $u_j^k = \min (\mathbf{u}^k)$\;
				\eIf{$u^k_j\geq 0$}{\textbf{Break;}}{Update $\mathbf{M}^k$ by removing the $j$-th row in $\mathbf{M}^k$ and go to \textbf{step 5}\;}}}{Obtain $\omega^k$ by solving \eqref{pro:3}\;
			Calculate $\mathbf{d}^{k+1}$ from \eqref{eq:37};
			
	}}
	\textbf{Output:} $\mathbf{d}$
	
\end{algorithm}
\renewcommand\baselinestretch{\linspread}\selectfont 

To summarize, the proposed RGPM-based algorithm is provided in Algorithm 1, and we analyze its computational complexity  as follows. The primary complexity of Algorithm 1  arises from calculating the gradient $\nabla\bar{f_a}(\mathbf{d}^k)$, whose complexity is on the order of  $\mathcal{O}(M_t^2n_1^2+M_t^2n_1n_2+M_t^2n_1n_3)$. Thus, the overall complexity of Algorithm 1 is $\mathcal{O}(KM_t^2n_1^2+KM_t^2n_1n_2+KM_t^2n_1n_3)$. In contrast, the complexity of a GA for problem \eqref{pro:2} is $\mathcal{O}(GNPM_t^2n_1^2+GNPM_t^2n_1n_2+GNPM_t^2n_1n_3)$ \cite{10.1007/s11042-020-10139-6}, where $G$, $N$, and $P$ denote the iteration number, population size, and chromosome length, respectively. To ensure convergence of the GA algorithm for solving this problem, we typically set $G$ on the order of $10^2$, $N$ and $P$ on the order of 
$10$. However, the proposed algorithm can achieve comparable performance with $K$ on the order of $10^2$, thus the proposed algorithm is computationally more efficient than GA.

\section{Numerical  results}
This section presents numerical results to validate our theoretical analysis on the minimum main lobe width and lower bounds of the ambiguity function, and evaluate the proposed algorithm's effectiveness in enhancing the MIMO radar performance. In our simulations, we consider an MA-enabled FH-MIMO radar system where the transmitter is equipped with $M_t = 8$ MAs and the receiver employs a conventional ULA with $M_r = 8$ antennas, spaced at half-wavelength intervals. The minimum distance between the MAs is set as $\frac{\lambda}{2}$. The system operates in the X-band with a center frequency of $f_c=8.2\mathrm{GHz}$ and bandwidth $B=8\mathrm{MHz}$. The PRI is set to $T_P=20\mu s$ and the pulse width is $T_w=6\mu s$. The FH code design follows the approach in \cite{eedaraOptimumCodeDesign2020}. For Algorithm 1, the convergence threshold and maximum iteration number are set to $T=10^{-2}$ and $K=150$, respectively. Other system parameters are as follow unless otherwise specified: $L/\lambda=7$, $Q=6$, $K=8$, $\Delta_f=1\mathrm{MHz}$, $\Delta_t=1\mathrm{\mu s}$, $f_{max}=10\mathrm{MHz}$, and a sampling frequency of $f_s = 160\mathrm{MHz}$.

\begin{figure}[t]
	\centering
	\includegraphics[width=0.4\textwidth]{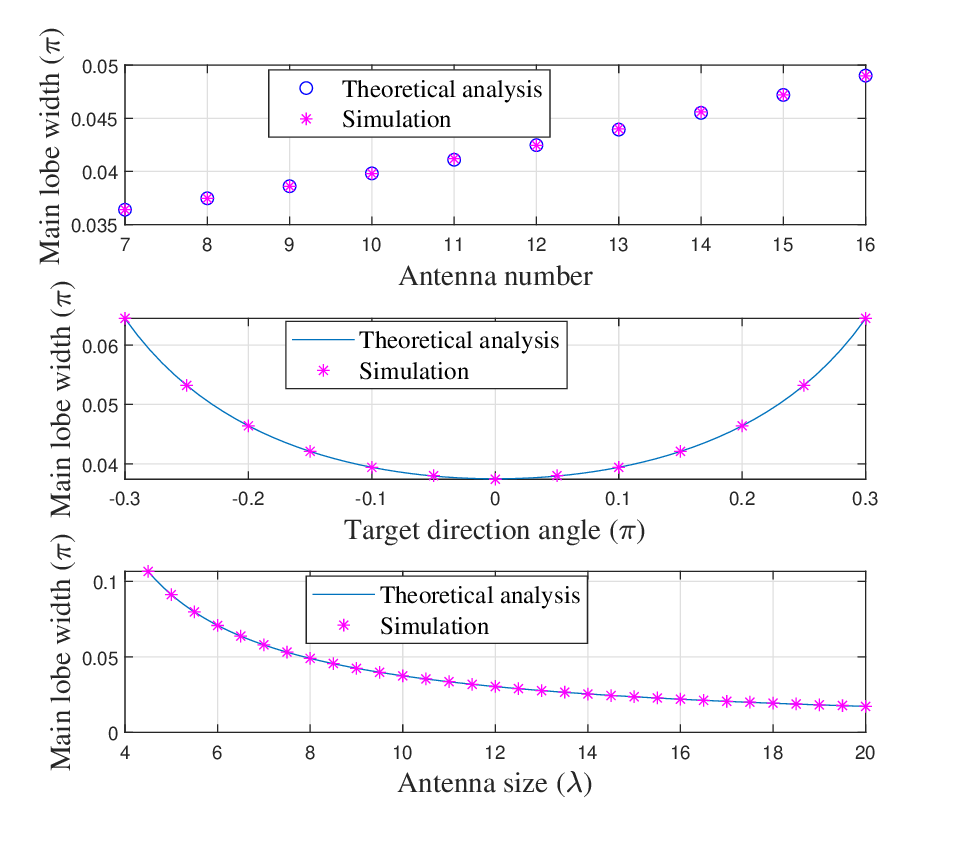}
	\vspace{-0.3cm}
	\caption{Minimum main lobe width versus antenna size $L$, antenna number $M_t$, and target direction angle $\theta$.}		
	\label{fig:figure22}  
	\vspace{-0.3cm}
\end{figure}

\begin{figure}[t]
	\centering
	\includegraphics[width=0.4\textwidth]{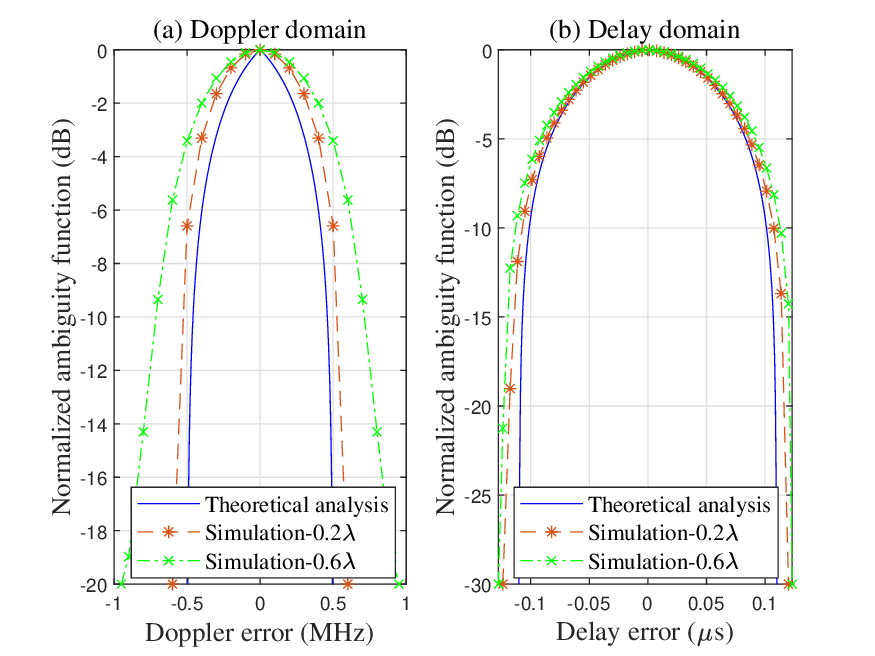}
	\caption{Lower bound of the ambiguity function in the Doppler and delay domains.}	
	\label{fig:.0723}
	\vspace{-0.3cm}
\end{figure}

First, in Fig. \ref{fig:figure22}, we validate our theoretical analysis of the relationship between $B_{min}$, $L/\lambda$, $M_t$, and $\theta$ as derived in Section \ref{section3-A}. We evaluate the main lobe width of the ambiguity function in the angular domain under the optimal antenna distribution \eqref{eq:0406_2} for different values of $L/\lambda$, $M_t$, $\theta$, and compare the results with those obtained from \eqref{eq:21}. The results show a strong agreement between the theoretical predictions and the numerical simulations for the minimum main lobe width, which validates that the expression in \eqref{eq:21} is accurate. This expression can therefore serve as a practical reference for system design under various antenna configurations.

\begin{figure}[t]
	\centering
	\includegraphics[width=0.4\textwidth]{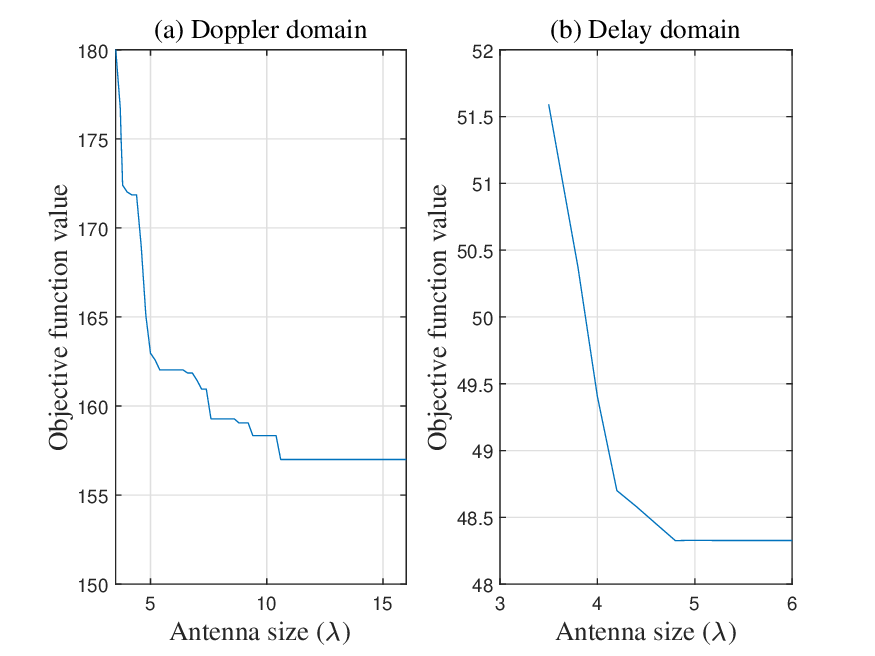}
	\caption{Objective function versus antenna size in the Doppler and delay domains.}	
	\label{fig:.0923}	 
\end{figure}

Next, we examine the lower bounds of the ambiguity function in both the Doppler and delay domains, as illustrated in Fig. \ref{fig:.0723}. As discussed in Section \ref{section3-B}, the term $e^{j2\pi J_{m,m',\theta}(\mathbf{d})}$ exhibits a periodic behavior, thus to simulate such a behavior, we set  $d_{t,i}\in[\frac{\lambda}{2},\frac{3\lambda}{2}]$, which ensures that $e^{j2\pi J_{m,m',\theta}(\mathbf{d})}$ has periodic variations. By discretizing $\{d_{t,i}\}$ with a step size of $0.2\lambda$ and $0.6\lambda$ and iterating through all possible values of $\{d_{t,i}\}$, we obtain the simulated lower bound of the ambiguity function. It can be observed from Fig. \ref{fig:.0723} that the main lobe width of the theoretical lower bound is narrower than that of the simulated lower bound, and their gap decreases as the step size decreases, which validate the  accuracy of our theoretical analysis on Theorems 2 and 3.


\begin{figure}[t]
	\centering
	\includegraphics[width=0.4\textwidth]{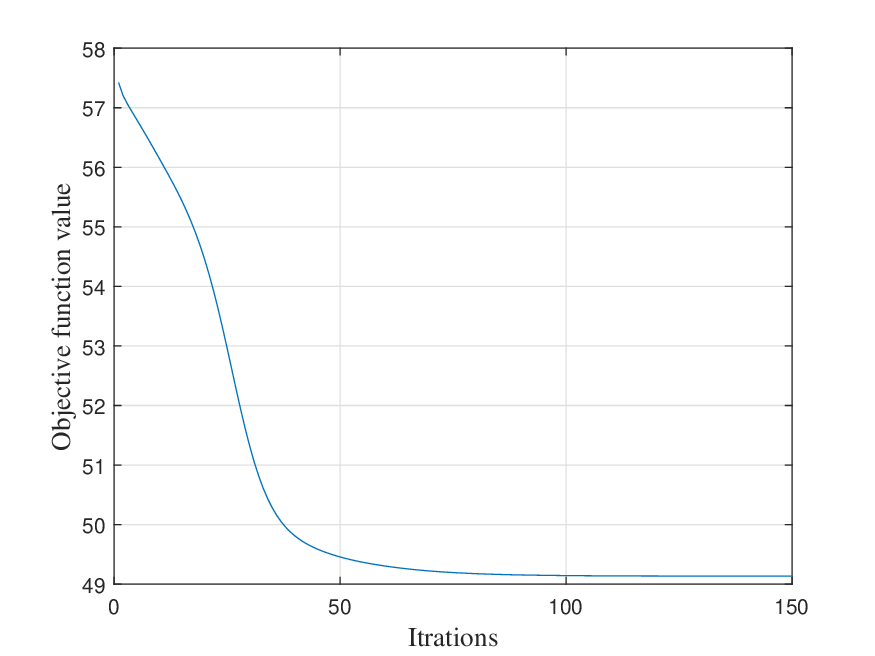}
	\caption{Convergence behavior of Algorithm 1.}	
	\label{figure:fig.0628_1}	 
	\vspace{-0.3cm}
\end{figure}

Fig. \ref{fig:.0923} illustrates the impact of the antenna size $L$ on  the objective function values in the Doppler and delay domains. It is seen that the objective function reaches a stable value when $L$ exceeds $4.8\lambda$ in the Doppler domain and $10.6\lambda$ in the delay domain, which is consistent with the discussion in Section \ref{section3-B}. This phenomenon suggests that increasing the antenna size too much does not yield additional performance gains.

In Fig. \ref{figure:fig.0628_1}, we illustrate the convergence behavior of the proposed RGPM-based algorithm by plotting the objective function value versus the number of iterations with $\alpha_1=0$, $\alpha_2=0$ and $\alpha_3=1$. From Fig. \ref{figure:fig.0628_1}, we can observe that the objective function value decreases monotonically and
converges to a stable value after about 50 iterations.

In Fig. \ref{figure:fig.12}, we investigate the radar performance gain  provided by MAs in the angular domain, where we set $\alpha_1=1$, $\alpha_2=0$, $\alpha_3=0$. Here, we compare the proposed algorithm with the GA \cite{10.1007/s11042-020-10139-6} and two benchmark antenna configurations: (1) the equidistant scheme with $d_{t,i}=\frac{\lambda}{2},i=1,2,\cdots,M_t-1$ that the existing system follows \cite{9969893}, and (2) the MMLWD scheme as specified in \eqref{eq:0406_2}. From Fig. \ref{figure:fig.12}, we observe that both  the proposed algorithm and GA achieve narrower main lobe widths and lower side lobe levels than the equidistant scheme. Additionally, the main lobe width attained by the proposed algorithm is very close to the theoretical bound achieved by the MMLWD scheme (given in \eqref{eq:21}), but the side lobe levels of the former is much lower than the latter. These results suggest that the proposed algorithm is able to effectively balance the  main lobe and side lobe performance, achieving near-optimal performance with low complexity.
\begin{figure}[t]
	\centering
	\includegraphics[width=0.4\textwidth]{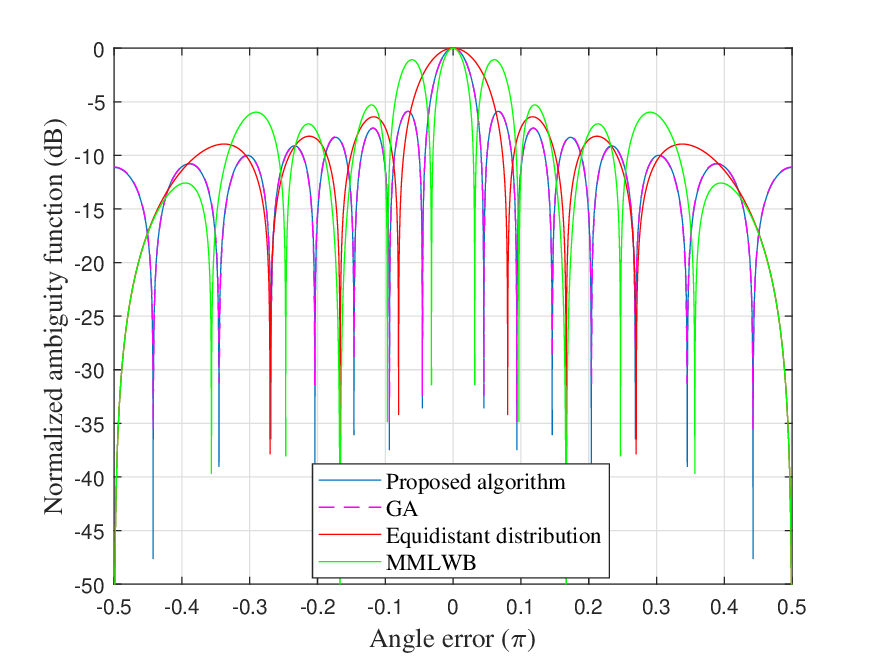}
	\caption{Ambiguity function in the angular domain with different antenna distributions.}		
	\label{figure:fig.12}	 
\end{figure}
\begin{figure}[t]
	\centering
	\includegraphics[width=0.4\textwidth]{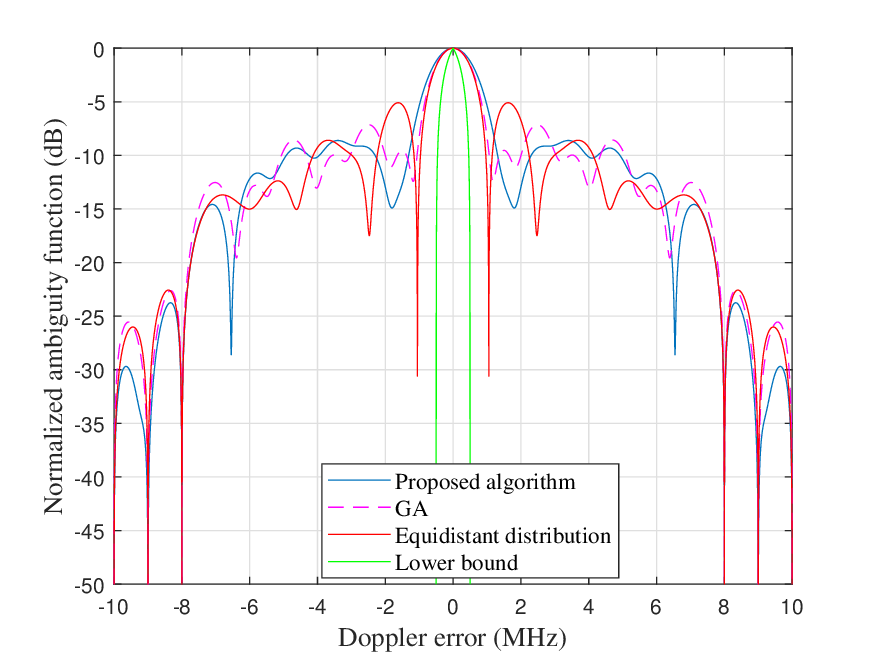}
	\caption{Ambiguity function in the Doppler domain with different antenna distributions.}	
	\label{figure:fig.4}	
	\vspace{-0.3cm}
\end{figure}
\begin{figure}[t]
	\centering
	\includegraphics[width=0.4\textwidth]{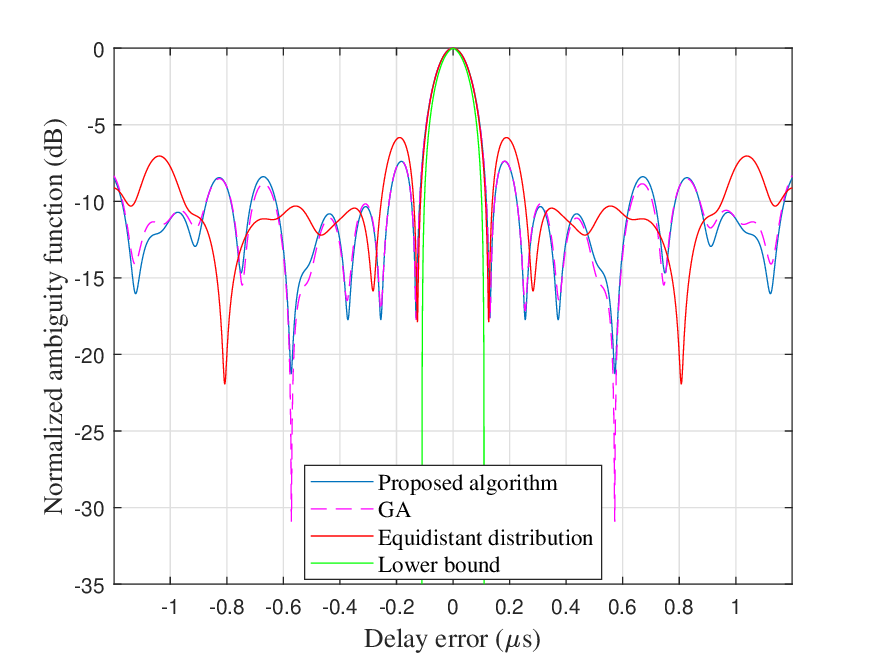}
	\caption{Ambiguity function in the delay domain with different antenna distributions.}	
	\label{figure:fig.3}	
\end{figure}
\begin{figure}[t]
	\centering
	\includegraphics[width=0.4\textwidth]{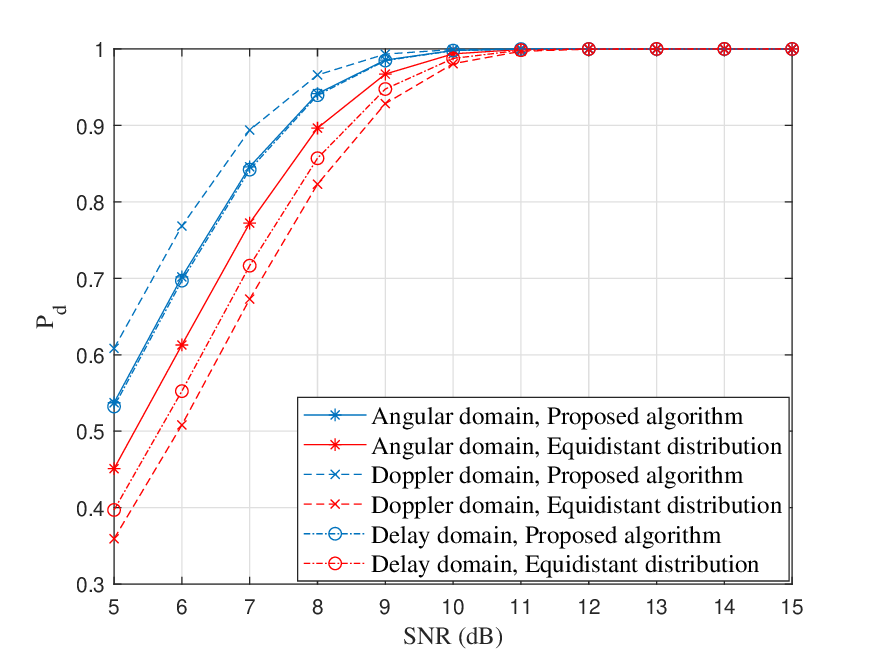}
	\caption{Detection probability versus SNR.}
	\label{figure:fig.0122}	
\end{figure} 
\begin{figure}[t]
	\centering
	\includegraphics[width=0.41\textwidth]{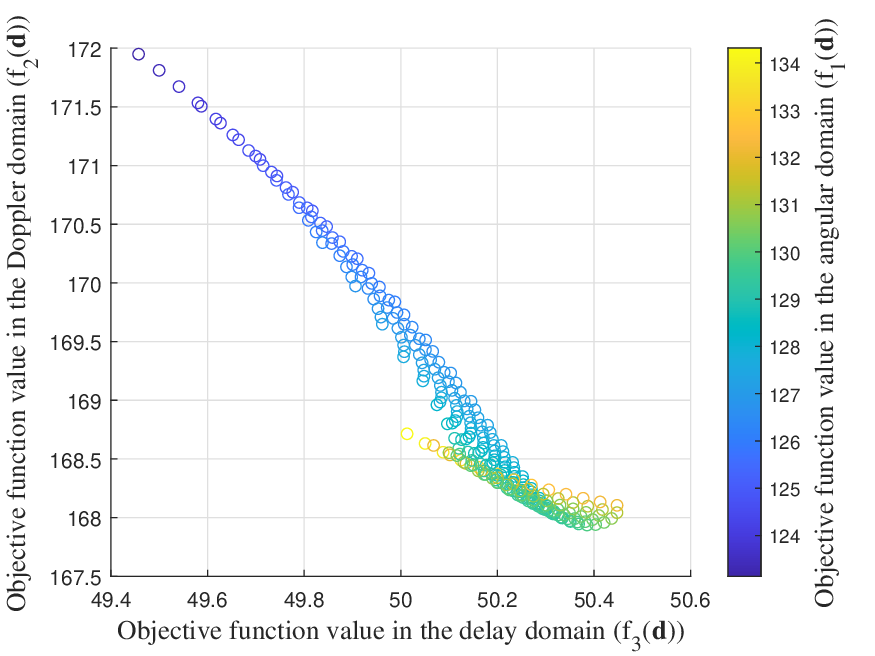}
	\caption{Objective function values achieved by different weighting coefficients.}	
	\label{figure:fig.0626}	
	\vspace{-0.3cm}
\end{figure}

Fig. \ref{figure:fig.4} and Fig. \ref{figure:fig.3} illustrate the ambiguity function in the Doppler and delay domains, under different antenna distributions, where we set ($\alpha_1=0$, $\alpha_2=1$, $\alpha_3=0$), and ($\alpha_1=0$, $\alpha_2=0$, $\alpha_3=1$), respectively. As mentioned in Section \ref{section3-B}, when $\theta = 0$, the MAs provide minimal performance gain in the Doppler and delay domains. Therefore, we focus on the performance gain achieved by the MAs when $\theta = \frac{\pi}{3}$. Here, we compare the proposed algorithm with the GA, and the equidistant scheme with $d_{t,i}=\frac{\lambda}{2},i=1,2,\cdots,M_t-1$ and the lower bound proposed in Section \ref{section3-B} are regarded as benchmarks. From Fig. \ref{figure:fig.4}, we observe that both the proposed algorithm and GA achieve significantly lower side lobe levels than the equidistant scheme with comparable main lobe width. This suggests that the proposed algorithm can improve the side lobe performance significantly in the Doppler domain with low complexity. Meanwhile, we can see that the main lobe width, both after optimization and under equidistant distribution, shows a noticeable deviation from the proposed lower bound. This discrepancy is primarily due to the relatively loose nature of the proposed lower bound. Additionally, in Fig. \ref{figure:fig.3}, we can see that both the proposed algorithm and GA achieve much lower side lobe levels than the equidistant scheme with the same main lobe width, which is very close to the lower bound. These observations indicate  that optimizing the antenna distribution has a limited effect on enhancing the main lobe performance of the ambiguity function in the Doppler and delay domains. However, it plays a significant role in improving its side lobe characteristics.

In Fig. \ref{figure:fig.0122}, we investigate the detection probability $P_d$ under a constant false alarm probability  $P_{fa}=0.01\%$. Here, we compare the proposed algorithm with  the equidistant scheme with $d_{t,i}=\frac{\lambda}{2}, i=1,2...,M_t-1$, and both systems employ the same FH code.  As we can see, the proposed algorithm achieves higher detection probability than the equidistant scheme. This implies that optimizing the antenna distribution can also improve the detection performance of the FH-MIMO radar system.

Finally, we  explore the tradeoff among the objective functions in the angular, Doppler, and delay domains by varying the weight parameters, and the results are shown in Fig. \ref{figure:fig.0626}, where higher temperature means larger $f_3(\mathbf{d})$ and vice versa. 
First, we can see that as $f_1(\mathbf{d})$ decreases, $f_1(\mathbf{d})$ also decreases. This implies that optimizing the antenna distribution can improve the properties of the ambiguity function in the delay and angular domains simultaneously. However, we also observe that as $f_2(\mathbf{d})$ decreases, both $f_3(\mathbf{d})$ and $f_1(\mathbf{d})$ will increase. This phenomenon indicates that the sensing resolution in the Doppler domain is in conflict with those in the delay and angular domains. 
Therefore, in practice, optimizing the distribution of the MA array to enhance the velocity resolution performance of FH-MIMO radar for target detection inevitably requires a tradeoff by compromising the resolution performance in terms of target range and angle. This suggests that to address various sensing requirements, the weight parameters corresponding to $f_j(\mathbf{d})$, $j=1,2,3$ should be properly adjusted to achieve a balanced performance across different domains. Such a limitation underscores the inflexibility of uniformly spaced fixed antenna arrays, as their rigid configuration fails to accommodate the varying sensing requirements, making it challenging to achieve better performance across diverse sensing domains.

\section{conclusion}
In this paper, we proposed a novel MA-enabled FH-MIMO radar system to improve the radar performance by exploiting the DoFs offered by antenna position optimization. In particular, we analyzed the relationship between the main lobe width of the ambiguity function in the angular domain and the antenna distribution. The expression of the theoretically minimum main lobe width was derived as a function of the target direction angle, the  antenna size and antenna number. Also, the optimal antenna distribution to achieve this minimum width was identified. Concurrently, we established lower bounds for the ambiguity function in  the Doppler and delay domains. To balance the main lobe and side lobe performance, we formulated an optimization problem for antenna distribution optimization and proposed a low-complexity RGPM-based algorithm to address it.
Numerical results showed that  the ambiguity function of our proposed system has lower side lobe levels and narrower main lobe width than that with conventional FPAs, and the accuracy and effectiveness of the proposed analysis and algorithm were also demonstrated.

 Future research  may consider the following aspects: 1. extending our research to a more general case where both transmit and receive antennas are movable; 2. designing MA-enabled DFRC/ISAC systems based on FH-MIMO radar; 3. designing an adaptive and reconfigurable hardware architecture that supports MA-enabled DFRC/ISAC systems.

{\appendix
\subsection{Proof of Theorem 2}\label{A.A_1}
	By assuming $\tau=0$ and $\theta=\theta'$, the ambiguity function \eqref{eq:15} can be simplified to
	\begin{equation}\label{eq:0611_3}
		\smaller
		\begin{aligned}
			\chi(0,v,\theta,\theta)=\!\!\!\sum_{m,m'=0}^{M_t-1}\!\kappa_{m,m'}(v)e^{j2\pi J_{m,m',\theta}(\mathbf{d})}	
		\end{aligned},
	\end{equation}
where
\begin{equation}\label{eq:1120_1}
	\smaller
	J_{m,m',\theta}(\mathbf{d})\triangleq(\sum_{i=0}^{m}d_{t,i}-\sum_{i=0}^{m'}d_{t,i})\sin\theta/\lambda,
\end{equation}
 and we define that 
\begin{equation}\label{eq:1120_2}
	\smaller
   \kappa_{m,m'}(v)\triangleq\sum_{q=0}^{Q-1}e^{j\pi\epsilon_{m,m',q}(v)}{\rm sinc}(\epsilon_{m,m',q}(v)),
 \end{equation}
we define that $\epsilon_{m,m',q}(v)\triangleq v\Delta_t-c_{m,q}+c_{m',q}$. When $m=m'$, we have $\epsilon_{m,m,q}(v)=v\Delta_t$, which means that $\kappa_{1,1}(v)e^{j2\pi J_{1,1,\theta}(\mathbf{d})}=\kappa_{i,i}(v)e^{j2\pi J_{i,i,\theta}(\mathbf{d})}$, where $i=1,2,\cdots,M_t-1$. Then, \eqref{eq:0611_3} can be further simplified to
	\begin{equation}\label{eq:0611_4}
		\smaller
		\begin{aligned}
			&\chi(0,v,\theta,\theta)= \varsigma_{v,\theta}(\mathbf{d})+M_t e^{j2\pi v\Delta t}{\rm sinc}(v\Delta t),
		\end{aligned}
	\end{equation}
	where 
	\begin{equation}\label{eq:0712_1}
		\smaller
		\varsigma_{v,\theta}(\mathbf{d}) =\!\!\!\!\sum_{\substack{m,m'=0\\m\neq m'}}^{M_t-1}\!\!\!\!\kappa_{m,m'}(v)e^{j2\pi J_{m,m',\theta}(\mathbf{d})}.
	\end{equation}
	Next, based on the triangle inequality and assuming that $L$ is large enough, we obtain
	\begin{equation}\label{eq:0611_5}
		\smaller
		\begin{aligned}
			&\lvert\chi(0,v,\theta,\theta)\lvert \geq \lvert M_t{\rm sinc}(v\Delta t)\lvert-\lvert\varsigma_{v,\theta}(\mathbf{d})\lvert,
		\end{aligned}
	\end{equation}
	where the equality holds if and only if the phases of $\varsigma_{v,\theta}(\mathbf{d})$ and $M_t e^{j2\pi v\Delta t}{\rm sinc}(v\Delta t)$ are opposite. From \eqref{eq:0611_5}, it is seen that the lower bound of $\lvert\chi(0,v,\theta,\theta)\lvert$ depends on the maximum value of $\lvert\varsigma_{v,\theta}(\mathbf{d})\lvert$. Besides, from \eqref{eq:0712_1}, we can see that $\varsigma_{v,\theta}(\mathbf{d})$ is the sum of $M_t^2-M_t$ complex numbers, and its maximum value can be obtained by resorting to the following triangle inequality
	\begin{equation}\label{eq:0611_6}
		\small
		\begin{aligned}
			&\lvert\varsigma_{v,\theta}(\mathbf{d})\lvert \\ &\leq\sum_{\substack{m,m'=0\\m\neq m'}}^{M_t-1}\sum_{q=0}^{Q-1}\lvert e^{j\pi\epsilon_{m,m',q}(v)}{\rm sinc}(\epsilon_{m,m',q}(v))e^{j2\pi J_{m,m',\theta}(\mathbf{d})}\large\lvert \\&=\sum_{\substack{m,m'=0\\m\neq m'}}^{M_t-1}\sum_{q=0}^{Q-1}\lvert {\rm sinc}(\epsilon_{m,m',q}(v))\large\lvert=\Xi^*(v).
				\vspace{-0.2cm}
		\end{aligned}
	\end{equation}
Therefore, when $\lvert M_t{\rm sinc}(v\Delta t)\lvert\geq\Xi^*(v)$, $\lvert\chi(0,v,\theta,\theta)\lvert\geq\lvert M_t{\rm sinc}(v\Delta t)\lvert-\Xi^*(v)$ holds, otherwise, we have $\chi^*(v)= 0$. This thus completes the proof.
	\vspace{-0.1cm}
	
\subsection{Proof of Theorem 3}\label{A.A_2}
	By assuming $v=0$ and $\theta=\theta'$, the ambiguity function \eqref{eq:15} can be simplified into 
\begin{equation}\label{eq:0613_4}
	\smaller
	\begin{aligned}
		&\chi(\tau,0,\theta,\theta)\\
		&=\!\!\!\!
		\begin{aligned}[t]\sum_{m,m'\!=0}^{M_t\!-\!1}\sum_{q,q'\!=0}^{Q-1}&\chi^{r}(\tau\!-\!(q'\!-\!q){\Delta}_t,-(c_{m'\!,q'}-c_{m\!,q})\Delta_f)\\&e^{-j2{\pi\Delta}_fc_{m',q'}\tau}e^{j2\pi J_{m,m',\theta}(\mathbf{d})}.
		\end{aligned}
	\end{aligned}
\end{equation}
Then, we can see that when $m=m'$, $J_{m,m',\theta}(\mathbf{d})=0$ holds, which implies that when $m=m'$, the antenna distribution will not affect ambiguity function. Therefore, depending on whether $m$ is equal to $m'$ or not, we transform \eqref{eq:0613_4} into $\chi(\tau,0,\theta,\theta)= \Lambda^d(\tau)+\Upsilon^d(\tau)$, where
\begin{equation}\label{eq:0613_5}
	\smaller
	\begin{aligned}
		&\Lambda^d(\tau)	\\&\begin{aligned}=\sum_{\substack{m,m'=0\\m\neq m'}}^{M_t\!-\!1}\sum_{q,q'\!=0}^{Q-1}
			&\!\chi^{r}(\tau\!-\!(q'\!-\!q){\Delta}_t,(c_{m,q}-c_{m',q'})\Delta_f)\\&e^{-j2{\pi\Delta}_fc_{m',q'}\tau}e^{j2\pi J_{m,m',\theta}(\mathbf{d})}.
		\end{aligned}
	\end{aligned}
\end{equation}
Next, we can see from \eqref{eq:0613_5} that the magnitude and phase of $\Lambda^d(\tau)$ are determined by the antenna distribution, while $\Upsilon^d(\tau)$ is a constant.
Based on the triangle inequality and assuming that $L$ is large enough, the minimum value of $\lvert\chi(\tau,0,\theta,\theta)\lvert$ is attained when $\Upsilon^d(\tau)$ and $\Lambda^d(\tau)$ are opposite, i.e,
$\chi(\tau,0,\theta,\theta)\geq \lvert\Upsilon^d(\tau)\lvert-\lvert\Lambda^d(\tau)\lvert$. Then, the maximum value of $\lvert\Lambda^d(\tau)\lvert$ can be obtained via
\begin{equation}\label{eq:0613_6}
	\smaller
	\begin{aligned}
		&\lvert\Lambda^d(\tau)\lvert\\&\leq\!\!\!\sum_{\substack{m,m'=0\\m\neq m'}}^{M_t\!-\!1}\sum_{q,q'\!=0}^{Q-1}
		\!\lvert\chi^{r}(\tau\!-\!(q'\!\!-\!q){\Delta}_t,(c_{m,q}\!-c_{m',q'})\Delta_f)\\&\qquad e^{j2\pi J_{m,m',\theta}(\mathbf{d})}e^{-j2{\pi\Delta}_fc_{m',q'}\tau}\lvert\\&=\!\sum_{\substack{m,m'=0\\m\neq m'}}^{M_t\!-\!1}\!\sum_{q,q'\!=0}^{Q-1}\!\lvert\chi^{r}(\tau\!-\!(q'\!-\!q){\Delta}_t,(c_{m,q}\!-c_{m',q'})\Delta_f)\lvert\\&=\Xi^d(\tau).
	\end{aligned}
\end{equation}
Therefore, we can obtain that when $\lvert\Upsilon^d(\tau)\lvert\geq\Xi^d(\tau)$, $\lvert\chi(\tau,0,\theta,\theta)\lvert\geq\lvert\Upsilon^d(\tau)\lvert-\Xi^d(\tau)$. Otherwise,  we have $\lvert\chi(\tau,0,\theta,\theta)\lvert\geq0$. This thus completes the proof.
	
\subsection{The formulation of $\lvert\chi(\tau,v,\theta,\theta')\lvert^2$}\label{A.A}
We define the component of $\chi(\tau,v,\theta,\theta')$ along the x-axis as $\chi(\tau,v,\theta,\theta')_x$
and along the y-axis as $\chi(\tau,v,\theta,\theta')_y$. According to the Pythagorean theorem, we have
\begin{equation}\label{eq:1023_9}
	\small
	\lvert\chi(\tau,v,\theta,\theta')\lvert^2=\chi(\tau,v,\theta,\theta')_x^2+\chi(\tau,v,\theta,\theta')_y^2,
\end{equation}
where
\begin{equation}\label{eq:1023_10}
	\small
	\begin{aligned}
			&\chi(\tau,v,\theta,\theta')_x\\&\begin{aligned}
				=\sum_{m,m'\!=0}^{M_t\!-\!1}\sum_{q,q'\!=0}^{Q-1}&\epsilon(\tau\!-\!(q'\!-\!q){\Delta}_t,v\!-\!(c_{m'\!,q'}-c_{m\!,q})\Delta_f)\\&\cos(\zeta_{m,m',q,q'}(\mathbf{h}))
			\end{aligned}
	\end{aligned}
\end{equation}
and
\begin{equation}\label{eq:1023_11}
	\small
		\begin{aligned}
		&\chi(\tau,v,\theta,\theta')_y\\&\begin{aligned}
			=\sum_{m,m'\!=0}^{M_t\!-\!1}\sum_{q,q'\!=0}^{Q-1}&\epsilon(\tau\!-\!(q'\!-\!q){\Delta}_t,v\!-\!(c_{m'\!,q'}-c_{m\!,q})\Delta_f)\\&\sin(\zeta_{m,m',q,q'}(\mathbf{h})),
		\end{aligned}
	\end{aligned}
\end{equation}
and we define that $\mathbf{h}\triangleq[\tau,v,\theta,\theta',\mathbf{d}^T,m]^T$, 
\begin{equation}\label{eq:0919_4}
	\small
	\epsilon(\tau,v)\triangleq\left\{\begin{aligned}
		&	\frac{{\Delta}_t\! -\!\!  \lvert\tau\lvert}{{\Delta}_t}{\rm sinc}(v(\Delta_ t\! -\! \lvert\tau\lvert)),\quad \lvert \tau\lvert\leq \Delta_t\\
		&0,\quad\quad\quad\quad\quad\quad\quad\quad\quad\quad\quad\quad\,  \mathrm{otherwise}
	\end{aligned}
	\right.,
\end{equation}
and
\begin{equation}\label{eq:0716_6}
	\small
	\begin{aligned}
		&\zeta_{m,m',q,q'}(\mathbf{h})\\&\begin{aligned}
			\triangleq&\pi(v\!-(c_{m'\!,q'}-c_{m\!,q}) \Delta_f)(\Delta_t -\! \tau\!+\!(q'\!-\!q){\Delta}_t)\\&-2\pi\Delta_fc_{m',q'}\tau+2\pi(\sum_{i=0}^{m}d_{t,i}\sin(\theta)-\sum_{i=0}^{m'}d_{t,i}\sin(\theta'))/\lambda.
		\end{aligned}
	\end{aligned}
\end{equation}
\subsection{The expression of $\frac{\partial\lvert\chi(\tau,v,\theta,\theta')\lvert^2}{\partial d_{t,x}}$}\label{A.B}
According to the rule of partial derivatives, we have
\begin{equation}\label{eq:1023_6}
	\small
	\begin{aligned}
	\frac{\partial\lvert\chi(\tau,v,\theta,\theta')\lvert^2}{\partial d_{t,x}}=&
		2\bigg(\chi(\tau,v,\theta,\theta')_x\frac{\partial\chi(\tau,v,\theta,\theta')_x}{\partial d_{t,x}}\\&+\chi(\tau,v,\theta,\theta')_y\frac{\partial\chi(\tau,v,\theta,\theta')_y}{\partial d_{t,x}}\bigg),
\end{aligned}
\end{equation}
where
\begin{equation}\label{eq:1023_7}
	\small
	\begin{aligned}
		&\frac{\partial\chi(\tau,v,\theta,\theta')_x}{\partial d_{t,x}}\\&\begin{aligned}
		=-\!\!\!\!\!\!\sum_{m,m'\!=0}^{M_t\!-\!1}\sum_{q,q'\!=0}^{Q-1}&\epsilon(\tau\!-\!(q'\!-\!q){\Delta}_t,v\!-\!(c_{m'\!,q'}-c_{m\!,q})\Delta_f)\\&\sin(\zeta_{m,m',q,q'}(\mathbf{h}))\frac{\partial\zeta_{m,m',q,q'}(\mathbf{h})}{\partial d_{t,x}}
		\end{aligned}
	\end{aligned}
\end{equation}
and
\begin{equation}\label{eq:1023_8}
	\small
	\begin{aligned}
	&\frac{\partial\chi(\tau,v,\theta,\theta')_y}{\partial d_{t,x}}\\&\begin{aligned}
		=\sum_{m,m'\!=0}^{M_t\!-\!1}\sum_{q,q'\!=0}^{Q-1}&\epsilon(\tau\!-\!(q'\!-\!q){\Delta}_t,v\!-\!(c_{m'\!,q'}-c_{m\!,q})\Delta_f)\\&\cos(\zeta_{m,m',q,q'}(\mathbf{h}))\frac{\partial\zeta_{m,m',q,q'}(\mathbf{h})}{\partial d_{t,x}}.
	\end{aligned}
\end{aligned}
\end{equation}
From \eqref{eq:0716_6}, we can obtain that $\frac{\partial\zeta_{m,m',q,q'}(\mathbf{h})}{\partial d_{t,x}}$ can be expressed as
\begin{equation}\label{eq:0520_3}
	\small
	\frac{\partial\zeta_{m,m',q,q'}(\mathbf{h})}{\partial d_{t,x}}={2\pi}(\gamma(x,m)\sin\theta-\gamma(x,m')\sin\theta')/\lambda,
\end{equation}
where
\begin{equation}\label{eq:0520_4}
	\small
	\gamma(x,m)\triangleq \left\{\begin{aligned}
		&1,\quad x\geq m\\
		&0,\quad x<m
	\end{aligned}
	\right
	..
\end{equation}
}

\vspace{-0cm}
\bibliographystyle{IEEEtran} 
\bibliography{ref1.bib}

\begin{thebibliography}{10}
\providecommand{\url}[1]{#1}
\csname url@samestyle\endcsname
\providecommand{\newblock}{\relax}
\providecommand{\bibinfo}[2]{#2}
\providecommand{\BIBentrySTDinterwordspacing}{\spaceskip=0pt\relax}
\providecommand{\BIBentryALTinterwordstretchfactor}{4}
\providecommand{\BIBentryALTinterwordspacing}{\spaceskip=\fontdimen2\font plus
\BIBentryALTinterwordstretchfactor\fontdimen3\font minus
  \fontdimen4\font\relax}
\providecommand{\BIBforeignlanguage}[2]{{%
\expandafter\ifx\csname l@#1\endcsname\relax
\typeout{** WARNING: IEEEtran.bst: No hyphenation pattern has been}%
\typeout{** loaded for the language `#1'. Using the pattern for}%
\typeout{** the default language instead.}%
\else
\language=\csname l@#1\endcsname
\fi
#2}}
\providecommand{\BIBdecl}{\relax}
\BIBdecl

\bibitem{9656537}
K.~Wu, J.~A. Zhang, X.~Huang, and Y.~J. Guo, ``Frequency-hopping {MIMO}
  radar-based communications: An overview,'' \emph{IEEE Aerosp. Electron. Syst.
  Mag.}, vol.~37, no.~4, pp. 42--54, 2022.

\bibitem{9540344}
J.~A. Zhang, F.~Liu, C.~Masouros, R.~W. Heath, Z.~Feng, L.~Zheng, and
  A.~Petropulu, ``An overview of signal processing techniques for joint
  communication and radar sensing,'' \emph{IEEE J. Sel. Top. Signal Process.},
  vol.~15, no.~6, pp. 1295--1315, 2021.

\bibitem{shi_enhanced_2022}
Y.~Shi, K.~An, X.~Lu, and Y.~Li, ``Enhanced {index} {modulation}-{based}
  {frequency} {hopping}: {Resist} {power}-{correlated} {reactive} {jammer},''
  \emph{IEEE Wireless Commun. Lett.}, vol.~11, no.~4, pp. 751--755, Apr. 2022.

\bibitem{9266402}
Ailiya, W.~Yi, and Y.~Yuan, ``Reinforcement learning-based joint adaptive
  frequency hopping and pulse-width allocation for radar anti-jamming,'' in
  \emph{Proc. IEEE Radar Conf.}, 2020, pp. 1--6.

\bibitem{163565}
S.~Maric and E.~Titlebaum, ``A class of frequency hop codes with nearly ideal
  characteristics for use in multiple-access spread-spectrum communications and
  radar and sonar systems,'' \emph{IEEE Trans. Commun.}, vol.~40, no.~9, pp.
  1442--1447, 1992.

\bibitem{9427572}
K.~Wu, J.~A. Zhang, X.~Huang, Y.~J. Guo, and J.~Yuan, ``Reliable
  frequency-hopping {MIMO} radar-based communications with multi-antenna
  receiver,'' \emph{IEEE Trans. Commun.}, vol.~69, no.~8, pp. 5502--5513, 2021.

\bibitem{9969893}
F.~Qiu, M.-M. Zhao, L.~Li, and M.-J. Zhao, ``Improved information embedding for
  frequency hopping-based {MIMO} {DFRC} system,'' \emph{IEEE Wireless Commun.
  Lett.}, vol.~12, no.~2, pp. 346--350, 2023.

\bibitem{chenMIMORadarAmbiguity2008}
C.-Y. Chen and P.~P. Vaidyanathan, ``{MIMO} radar ambiguity properties and
  optimization using frequency-hopping waveforms,'' \emph{IEEE Trans. Signal
  Process.}, vol.~56, no.~12, pp. 5926--5936, Dec. 2008.

\bibitem{gogineniFrequencyHoppingCodeDesign2012}
S.~Gogineni and A.~Nehorai, ``Frequency-{hopping} {code} {design} for {MIMO}
  {radar} {estimation} {using} {sparse} {modeling},'' \emph{IEEE Trans. Signal
  Process.}, vol.~60, no.~6, pp. 3022--3035, Jun. 2012.

\bibitem{zhouFrequencyhoppingCodeOptimization2016}
X.~Zhou, H.~Wang, Y.~Cheng, Y.~Qin, and X.~Xu, ``Frequency-hopping code
  optimization for radar coincidence imaging by exploiting the dictionary
  matrix,'' in \emph{Proc. CIE Int. Conf. Radar, RADAR}, Oct. 2016, pp. 1--4.

\bibitem{hanJointlyOptimalDesign2016}
K.~Han and A.~Nehorai, ``Jointly optimal design for {MIMO} radar
  frequency-hopping waveforms using game theory,'' \emph{IEEE Trans. Aerosp.
  Electron. Syst.}, vol.~52, no.~2, pp. 809--820, Apr. 2016.

\bibitem{zhangDualFunctionMIMORadarCommunications2021}
H.~Zhang and J.~Zheng, ``Dual-{function} {MIMO} {radar}-{vommunications} via
  {group} {frequency} {hopping} {code} {selection} and {PSK} {modulation},'' in
  \emph{Proc. IEEE Radar. Conf.}, Dec. 2021, pp. 2994--2998.

\bibitem{eedaraOptimumCodeDesign2020}
I.~P. Eedara, M.~G. Amin, and A.~Hoorfar, ``Optimum {code} {design} {using}
  {genetic} {algorithm} in {frequency} {hopping} {dual} {function} {MIMO}
  {radar} {communication} {systems},'' in \emph{Proc. IEEE Radar. Conf.}, Sep.
  2020, pp. 1--6.

\bibitem{hassanienDualfunctionMIMORadarcommunications2017}
A.~Hassanien, B.~Himed, and B.~D. Rigling, ``A dual-function {MIMO}
  radar-communications system using frequency-hopping waveforms,'' in
  \emph{Proc. IEEE Radar. Conf.}, May 2017, pp. 1721--1725.

\bibitem{wangPhaseModulatedCommunications2020}
X.~Wang and A.~Hassanien, ``Phase {modulated} {communications} {embedded} in
  {correlated} {FH}-{MIMO} {radar} {waveforms},'' in \emph{Proc. IEEE Radar.
  Conf.}, Sep. 2020, pp. 1--6.

\bibitem{moffetMinimumredundancyLinearArrays1968}
A.~Moffet, ``Minimum-redundancy linear arrays,'' \emph{IEEE Trans. on Antennas
  Propag.}, vol.~16, no.~2, pp. 172--175, Mar. 1968.

\bibitem{vaidyanathanSparseSensingCoPrime2011}
P.~P. Vaidyanathan and P.~Pal, ``Sparse {sensing} {with} {co}-{prime}
  {samplers} and {arrays},'' \emph{IEEE Trans. Signal Process.}, vol.~59,
  no.~2, pp. 573--586, Feb. 2011.

\bibitem{palNestedArraysNovel2010}
P.~Pal and P.~P. Vaidyanathan, ``Nested {arrays}: {A} {novel} {approach} to
  {array} {processing} {with} {enhanced} {degrees} of {freedom},'' \emph{IEEE
  Trans. Signal Process.}, vol.~58, no.~8, pp. 4167--4181, Aug. 2010.

\bibitem{chenMinimumRedundancyMIMO2008}
C.-Y. Chen and P.~P. Vaidyanathan, ``Minimum redundancy {MIMO} radars,'' in
  \emph{Proc. IEEE Int. Symp Circuits Syst}, May 2008, pp. 45--48.

\bibitem{qinDOAEstimationMixed2014}
S.~Qin, Y.~D. Zhang, and M.~G. Amin, ``{DOA} estimation of mixed coherent and
  uncorrelated signals exploiting a nested {MIMO} system,'' in \emph{Proc. IEEE
  Benjamin Franklin Symp. Microw. Antenna Sub-Syst. Radar, Telecommun., Biomed.
  Appl., BenMAS}, Sep. 2014, pp. 1--3.

\bibitem{9264694}
K.-K. Wong, A.~Shojaeifard, K.-F. Tong, and Y.~Zhang, ``Fluid antenna
  systems,'' \emph{IEEE Trans. Wireless Commun.}, vol.~20, no.~3, pp.
  1950--1962, 2021.

\bibitem{wongFluidAntennaMultiple2022}
K.-K. Wong and K.-F. Tong, ``Fluid {antenna} {multiple} {access},'' \emph{IEEE
  Trans. Wireless Commun.}, vol.~21, no.~7, pp. 4801--4815, Jul. 2022.

\bibitem{10653737}
H.~Yang, H.~Xu, K.-K. Wong, C.-B. Chae, R.~Murch, S.~Jin, and Y.~Zhang,
  ``Position index modulation for fluid antenna system,'' \emph{IEEE Trans.
  Wireless Commun.}, pp. 1--1, 2024.

\bibitem{10694739}
F.~R. Ghadi, K.-K. Wong, F.~Javier López-Martínez, W.~K. New, H.~Xu, and
  C.-B. Chae, ``Physical layer security over fluid antenna systems: Secrecy
  performance analysis,'' \emph{IEEE Trans. Wireless Commun.}, pp. 1--1, 2024.

\bibitem{zhuMovableAntennasWireless2023}
L.~Zhu, W.~Ma, and R.~Zhang, ``Movable {antennas} for {wireless}
  {communication}: {Opportunities} and {challenges},'' \emph{IEEE Commun.
  Mag.}, pp. 1--7, 2023.

\bibitem{maCapacityMaximizationMovable2023}
W.~Ma, L.~Zhu, and R.~Zhang, ``Capacity {maximization} for {movable} {antenna}
  {enabled} {MIMO} {communication},'' in \emph{Proc. IEEE Int. Conf. Commun.},
  May 2023, pp. 5953--5958.

\bibitem{10497534}
Z.~Xiao, S.~Cao, L.~Zhu, Y.~Liu, B.~Ning, X.-G. Xia, and R.~Zhang, ``Channel
  estimation for movable antenna communication systems: A framework based on
  compressed sensing,'' \emph{IEEE Trans. Wireless Commun.}, vol.~23, no.~9,
  pp. 11\,814--11\,830, 2024.

\bibitem{10709885}
L.~Zhu, W.~Ma, Z.~Xiao, and R.~Zhang, ``Performance analysis and optimization
  for movable antenna aided wideband communications,'' \emph{IEEE Trans.
  Wireless Commun.}, pp. 1--1, 2024.

\bibitem{hejresNullSteeringPhased2004}
J.~Hejres, ``Null steering in phased arrays by controlling the positions of
  selected elements,'' \emph{IEEE Trans. Antennas Propag.}, vol.~52, no.~11,
  pp. 2891--2895, Nov. 2004.

\bibitem{zhuravlevExperimentalSimulationMultistatic2015}
A.~Zhuravlev, V.~Razevig, S.~Ivashov, A.~Bugaev, and M.~Chizh, ``Experimental
  simulation of multi-static radar with a pair of separated movable antennas,''
  in \emph{Proc. IEEE Int. Conf. Microwaves, Commun., Antennas Electron. Syst.,
  (COMCAS)}, Nov. 2015, pp. 1--5.

\bibitem{basbugDesignSynthesisAntenna2017}
S.~Basbug, ``Design and {synthesis} of {antenna} {array} {with} {movable}
  {elements} {along} {semicircular} {paths},'' \emph{IEEE Antennas Wirel.
  Propag. Lett.}, vol.~16, pp. 3059--3062, 2017.

\bibitem{ma_movable_2024}
W.~Ma, L.~Zhu, and R.~Zhang, ``Movable {antenna} {enhanced} {wireless}
  {sensing} via {antenna} {position} {optimization},'' \emph{IEEE Trans.
  Wireless Commun.}, vol.~23, no.~11, pp. 16\,575--16\,589, Nov. 2024.

\bibitem{10696953}
H.~Qin, W.~Chen, Q.~Wu, Z.~Zhang, Z.~Li, and N.~Cheng, ``Cramér-rao bound
  minimization for movable antenna-assisted multiuser integrated sensing and
  communications,'' \emph{IEEE Wireless Commun. Lett.}, vol.~13, no.~12, pp.
  3404--3408, 2024.

\bibitem{khalili_advanced_2024}
\BIBentryALTinterwordspacing
A.~Khalili and R.~Schober, ``Advanced {ISAC} {design}: {Movable} {antennas} and
  {accounting} for {dynamic} {RCS},'' Jul. 2024, arXiv:2407.20930 [eess].
  [Online]. Available: \url{http://arxiv.org/abs/2407.20930}
\BIBentrySTDinterwordspacing

\bibitem{9681340}
W.~Baxter, E.~Aboutanios, and A.~Hassanien, ``Joint radar and communications
  for frequency-hopped {MIMO} systems,'' \emph{IEEE Trans. Signal Process.},
  vol.~70, pp. 729--742, 2022.

\bibitem{levanon2004radar}
N.~Levanon and E.~Mozeson, \emph{Radar signals}.\hskip 1em plus 0.5em minus
  0.4em\relax NewYork: Wiley-IEEE Press, 2004.

\bibitem{sanantonioMIMORadarAmbiguity2007a}
G.~San~Antonio, D.~R. Fuhrmann, and F.~C. Robey, ``{MIMO} radar ambiguity
  functions,'' \emph{IEEE J. Sel. Top. Signal Process.}, vol.~1, no.~1, pp.
  167--177, Jun. 2007.

\bibitem{eedaraDualFunctionFrequencyHoppingMIMO2022}
I.~P. Eedara, M.~G. Amin, A.~Hoorfar, and B.~K. Chalise, ``Dual-function
  frequency-hopping {MIMO} radar system with {CSK} signaling,'' \emph{IEEE
  Trans. Aerosp. Electron. Syst.}, vol.~58, no.~3, pp. 1501--1513, Jun. 2022.

\bibitem{c2009radar}
P.~S. Jian~Li, \emph{MIMO radar signal processing}.\hskip 1em plus 0.5em minus
  0.4em\relax NewYork: Wiley-IEEE Press, 2009.

\bibitem{Armijo1966}
L.~Armijo, ``Minimization of functions having lipschitz continuous first
  partial derivatives,'' \emph{Pac. J. Math.}, vol.~16, no.~1, pp. 1--3, 1966.

\bibitem{10.1007/s11042-020-10139-6}
S.~Katoch, S.~S. Chauhan, and V.~Kumar, ``A review on genetic algorithm: past,
  present, and future,'' \emph{Multimedia Tools Appl.}, vol.~80, no.~5, p.
  8091–8126, 2021.

\end{thebibliography}

\end{document}